\documentclass{article}

\usepackage{arxiv}
\makeatletter
\renewcommand{\normalsize}{%
  \fontsize{12pt}{14pt}\selectfont  %
  \abovedisplayskip      7\p@ \@plus 2\p@ \@minus 5\p@
  \abovedisplayshortskip \z@ \@plus 3\p@
  \belowdisplayskip      \abovedisplayskip
  \belowdisplayshortskip 4\p@ \@plus 3\p@ \@minus 3\p@
}
\makeatother

\usepackage[utf8]{inputenc} %
\usepackage[T1]{fontenc}    %
\usepackage{url} %
\usepackage{booktabs} %
\usepackage{amsfonts} %
\usepackage{nicefrac} %
\usepackage{microtype} %
\usepackage{graphicx}
\usepackage{natbib}
\usepackage{doi}

\usepackage{multirow}
\usepackage{amsmath,amssymb,amsthm}%
\usepackage{hyperref} %
\usepackage{cleveref}%
\hypersetup{hidelinks}
\usepackage{mathrsfs}%
\usepackage{xcolor}%
\usepackage{textcomp}%
\usepackage{manyfoot}%
\usepackage{algorithm}%
\usepackage{algorithmicx}%
\usepackage{algpseudocode}%
\usepackage{listings}%
\usepackage[shortlabels]{enumitem}
\usepackage{comment}
\usepackage{bbm}
\usepackage{mathabx}

\RequirePackage{snapshot}

\def\cM{\mathcal{M}}
\def\cQ{\mathcal{Q}}
\def\cZ{\mathcal{Z}}
\def\cP{\mathcal{P}}

\def\cQpost{\mathcal{Q}^+}
\def\cZpost{\mathcal{Z}^+}

\def\pmin{p^{\downarrow}_{\cZ_h}}
\def\player{p^{\Updownarrow}_{\cZ_h}}

\def\ptheta{\varrho}

\let\oldpropto\propto
\renewcommand{\propto}{\mathrel{\,\oldpropto\,}}

\def\Xmin{X^\downarrow}
\def\Xmax{X^\uparrow}
\def\XminStar{X^{*\downarrow}}

\def\Xminpr{X'^{\downarrow}}
\def\Xpr{\widetilde{X}}
\def\xipr{\widetilde{\xi}}

\def\layer{X^{\Updownarrow}}
\def\layerpr{X'^{\Updownarrow}}

\def\Llayer{L_{\Updownarrow}}
\def\Ulayer{U_{\Updownarrow}}

\def\AuxRate{\widebar{M}}

\def\piY{\pi_Y}
\def\piM{\mu}

\def\G{{\psi \cup S}}
\def\Gtick{{\psi' \cup S}}

\def\gibbshmc{\texttt{GibbsHMC}}
\def\gibbsprior{\texttt{GibbsPrior}}
\def\gibbspost{\texttt{GibbsPost}}
\def\pmcmc{\texttt{pMCMC}}

\newcommand{\sort}[1]{\texttt{sort(}#1\texttt{)}}

\theoremstyle{plain}
\newtheorem{theorem}{Theorem}
\newtheorem{proposition}{Proposition}

\theoremstyle{definition}

\theoremstyle{remark}

\title{Exact Gibbs sampling for stochastic differential equations with gradient drift and constant diffusion}
\renewcommand{\shorttitle}{Gibbs sampling for constant diffusion SDEs}

\date{}

\author{%
	Xinyi Pei \thanks{Equal contribution}\\ 	\texttt{pxinyi@purdue.edu}
	\And
	Minhyeok Kim \footnotemark[\value{footnote}] \\	\texttt{kim4660@purdue.edu} \\~\\
	Department of Statistics, Purdue University \\
	West Lafayette, IN 47907, USA \\
	\And
	Vinayak Rao \thanks{Corresponding author}\\
	\texttt{varao@purdue.edu} \\
}

\begin{document}
\maketitle

\begin{abstract}
	Stochastic differential equations (SDEs) are an important class of time-series models, used to describe stochastic systems evolving in continuous time. Simulating paths from these processes, particularly after conditioning on noisy observations of the latent path, remains a challenge. Existing methods often introduce bias through time-discretization, require involved rejection sampling or debiasing schemes or are restricted to a narrow family of diffusions.
	In this work, we propose an exact Markov chain Monte Carlo (MCMC) sampling algorithm that is applicable to a broad subset of all SDEs with unit diffusion coefficient; after suitable transformation, this includes an even larger class of multivariate SDEs and most 1-d SDEs.
	We develop a Gibbs sampling framework that allows {\em exact} MCMC for such diffusions, without any discretization error.
	We demonstrate how our MCMC methodology requires only fairly straightforward simulation steps. Our framework can be extended to include parameter simulation, and allows tools from the Gaussian process literature to be easily applied.
	We evaluate our method on synthetic and real datasets, demonstrating superior performance to particle MCMC approaches.
\end{abstract}


\section{Introduction}\label{sec1}
Stochastic differential equations (SDEs)~\citep{oksendal2013stochastic} play an important role across a broad spectrum of fields; %
however, the accurate and efficient simulation of SDE paths poses significant challenges that remain only partially resolved. 
The primary difficulty with these models is that simulating from and evaluating state transition probabilities is an intractable problem. %
 Standard approaches introduce error through discrete-time approximations to the dynamics (giving for example the Euler-Maruyama method~\citep{kloeden1994numerical}), and can be extended to conditional simulation (given noisy observations) using standard discrete-time algorithms, such as particle filters~\citep{DelMoral1996} or particle MCMC~\citep{andrieu2010particle}.

In~\Cref{sec:ea}, we describe a series of {\em exact} sampling methods
~\citep{beskos2005exact, beskos2006retrospective, beskos2006exact, beskos2008factorisation} that avoid discretization error. These are all rejection sampling schemes where paths are proposed from a Brownian bridge and accepted or rejected based on a Poisson process realization. This ingenious simulation scheme, while exact, can be complex for some settings, and can have high rejection rates when the Brownian proposal deviates significantly from the SDE target.

In this paper, we develop a Markov chain Monte Carlo (MCMC) methodology that builds on the representation from these papers. %
Our methodology extends naturally to hidden Markov process settings where the SDE path is only observed noisily, and opens the way for practitioners to employ established MCMC methodologies like Hamiltonian Monte Carlo~\citep[HMC,][]{neal2011mcmc}. 
Our algorithm follows an auxiliary variable MCMC approach
outlined in~\citet{wang2020exact}. That algorithm
repeated two main steps: simulate an inhomogeneous Poisson point process conditioned on the diffusion trajectory, and simulate a diffusion trajectory given these Poisson events. 
It was, however, limited to a narrow class of EA1  diffusions, %
and here, we relax their restrictive assumptions, allowing MCMC inference for the much broader class EA2 and EA3 diffusions (defined later).
For these (and especially EA3 diffusions), the benefits of our MCMC approach are more apparent, allowing us to carefully order simulations to obtain algorithms much simpler than the rejection samplers they build upon. %

\subsection{Related work}
The exact rejection sampling methods~\citep{beskos2005exact, beskos2008factorisation, pollock2016exact}  above were originally designed for unconditional and end-point conditioned simulation.
In~\citet{beskos2006exact, sermaidis2013markov}, extensions to conditional simulation given noisy observations were studied, which proceed by partitioning the time interval at the observation times, and running the rejection sampler on each segment, together with a step to update path values at observation times. 
These require multiple rejection sampling steps as part of each MCMC iteration, %
and when applied to EA3 diffusions~\citep{sermaidis2013markov}, require the complicated EA3 rejection sampler. By  contrast, our MCMC approach is easier to implement, and can use sophisticated likelihood driven proposals (such as Hamiltonian Monte Carlo) to update the entire path simultaneously.

~\cite{fearnhead2008particle} extended the above ideas to HMM settings via a random-weight particle filter, %
however this still requires the complicated EA3 sampling methods that our approach avoids. \cite{fearnhead2010random} does avoid this through a clever Wald-type construction. %
However, these particle filtering methods 
focus on the filtering (rather than the smoothing) distribution, and 
can suffer from particle degeneracy issues. %
Some of these limitations can be tackled (at further computational cost) by incorporating them into particle MCMC (pMCMC) methods. In our experiments, we compare against time-discretized approximate pMCMC methods, and for one example, against an exact pMCMC method. 

A different approach was taken in~\citet{rhee2015unbiased}, who used multilevel Monte Carlo to allow unbiased estimation of functionals of diffusion paths. ~\citet{chada2021unbiased} used a similar idea, together with particle filtering to allow unbiased estimation in the posterior setting. While applicable to much more general SDEs than our approach, these are more complicated than our Gibbs scheme.

As already stated, our work builds on~\citet{wang2020exact}, that develops an MCMC algorithm that has the EA1 diffusion  as its stationary distribution. %
Extending this directly to EA2 and EA3 diffusions involves working with  non-Gaussian diffusions, and would naively involve intractable likelihoods and gradients. In this paper, we show how with some thought, we can have EA2 and EA3 MCMC samplers that are not much more complex than~\citet{wang2020exact}, but are simpler and more flexible than the EA2 and EA3 rejection sampling algorithms.

\section{Background}
\label{Background}

We will consider diffusions on the interval $[0,T]$  satisfying the SDE:
\begin{align}
dX_t = \alpha_{\theta}(X_t)\,dt + dW_t,  \quad X_0 \sim h^0; \quad t \in [0,T]. \label{eq:sde} %
\end{align} %
Here $X_t$ and $W_t$ are the state of the diffusion and a Brownian motion path at time $t$ respectively, $\alpha_{\theta}(x) = \nabla_x A_\theta(x)$ is the \emph{drift} term of \emph{gradient form}, and \(\theta \in \Theta\) are the process parameters. Informally,~\cref{eq:sde} describes the continuous evolution of \(X_t\) over time, driven by deterministic and stochastic components. \(X_t\) can be multidimensional, with \(X_t, W_t\) and \(\alpha_{\theta}(\cdot) \in \mathbb{R}^d\), %
though our focus is mostly on $d=1$.

For a general SDE, $\alpha_\theta(\cdot)$ can be more general, and the second term on the right of~\cref{eq:sde} is actually \(\sigma_{\theta}(\cdot)\,dW_t\), with a state-dependent \emph{diffusion term}.
For one-dimensional diffusions, which are our main focus, our assumptions involve no real loss of generality (see~\citet{moller2010state, pollock2016exact} for details on the Lamperti transform). In higher dimensions, this assumption is a limitation of the rejection sampling algorithms of Beskos and collaborators that our method inherits. Nevertheless, \cref{eq:sde} represents a broad and useful class of SDEs. 

\subsection{Notation} \label{sec:notation}
Our primary goal is to simulate SDE paths $X$ over an interval $[0,T]$, possibly conditioning on noisy observations $Y$ at a finite set of times $S$. Without loss of generality, $S$ will include the start and end times.
We use $\cQ$ to refer to the law over paths implied by~\cref{eq:sde}, and $\cQpost$ to refer to the law over paths conditioned on $Y$. We call these the {\em prior} and {\em posterior} over paths.
For most of this paper, we assume the data $Y$ and the parameters $\theta$ are fixed and known, and make them implicit to our notation.
When necessary, we will include $\theta$-subscripts (e.g.\ $\cQ_\theta$ and  $\cQpost_\theta$). 

For any $\tau \subset [0,T]$, we write $X_\tau$ for $X$ at times $t \in \tau$. %
We use $\Xmin = (\Xmin_v,\Xmin_t) \in \mathbb{R} \times [0,T]$ %
for the value and time of the minimum of $X$, and $\Xmax=(\Xmax_v,\Xmax_t)$ for the maximum. %
The \emph{layer} $\layer$ of $X$ is a pair $(\Llayer,\Ulayer) \in \mathbb{R}^2$ with $\Llayer \le \Xmin \le \Xmax \le \Ulayer$ . %

 We write 
 $\cZ_h$ for the law of an \emph{$h$-biased Brownian bridge} on $[0,T]$; to sample a path from this, first simulate the endpoints $X_0$ and $X_T$ from a probability density $h(X_0,X_T)$, and then connect these with a Brownian bridge (see~\Cref{sec:ea}).
 We write $\cZ_h(\cdot\mid X_A)$ or $\cZ_h(dX\mid X_A)$ for the conditional law given path values at times $A$, and $p_{\cZ_h}(X_B\mid X_A)$ for the corresponding density at a finite set of times $B$. When $A$ includes the end-times $0$ and $T$, then for any realization of $X_A$, $\cZ_h(\cdot \mid X_A)$ is the law of a Brownian bridge (really a collection of independent Brownian bridges over the intervals defined by $A$), and $p_{\cZ_h}$ is just a $|B|$-dimensional Gaussian distribution. %
 Sometimes we will let $A$ and $B$ intersect (most commonly, $\{0,T\} = A \cap B$). In this case, we abuse notation slightly, and use $p_{\cZ_h}(X_B\mid X_A)$ to refer to the density of $X_{B \setminus A}$. %
 This is also the case when we write $\cZ_h(dX \mid X_A)$.
 We write $\cZ_h(dX\mid \Xmin)$ for the regular conditional probability of an $h$-biased Brownian bridge given its minimum $\Xmin$, and 
$\cZ_h(\cdot \mid \layer)$ for the conditional law given the layer $\layer$. The former corresponds to a Bessel process, while the latter is not a standard stochastic process. 
 Finally, for a sequence $A$, $\sort{A}$ refers to its elements in increasing order.

\vspace{-.1in}
\section{Exact Rejection Sampling Algorithms without discretization error} \label{sec:ea}
We highlight key points, for more details,~\citet{pollock2016exact} provides an excellent review. %
Recall, $\alpha_\theta(x) =\nabla_x A_\theta(x)$, and %
$h^0$ is the prior over the initial value $X_0$. We assume:  (a) $ \int_{-\infty}^\infty \int_{-\infty}^{\infty}h^0(u_1) \exp\left(A_{\theta}(u_2) - A_{\theta}(u_1)- \frac{(u_2-u_1)^2}{2T}\right)\,du_2\,du_1 := c_{\theta} < \infty \ \forall \theta$, and (b) there exists a lower bound $\alpha^\downarrow_{\theta} \leq \frac{1}{2}\left(\alpha_{\theta}^2(x) + \alpha_{\theta}^{\prime}(x)\right)  \  \forall x \in \mathbb{R}$. For more about these mild assumptions, see~\citet{pollock2016exact}; usually (a) is written with $u_1$ fixed to a {\em known} initialization of the path.
With (a), we can define a density \(h_{\theta}(u_1, u_2) = \frac{1}{c_{\theta}}h^0(u_1)  \exp\left(A_{\theta}(u_2) - A_{\theta}(u_1)- \frac{(u_2-u_1)^2}{2T}\right)\).
Assumption (b) gives %
\vspace{-.1in}
\begin{align}
  \phi_{\theta}(x) := \frac{1}{2}\left(\alpha_{\theta}^2(x) + \alpha_{\theta}^{\prime}(x)\right) - \alpha^\downarrow_{\theta} \ge 0.   \label{eq:phi}
\end{align}
Denote by $\cZ_h$ the law of a biased Brownian motion whose endpoints $(X_0,X_T)$ follow $h_\theta(\cdot,\cdot)$,  with the rest of the path conditionally following a Brownian bridge. One can use Girsanov's theorem to show~\citep{beskos2005exact} that the Radon-Nikodym derivative between the diffusion measure $\mathcal{Q}$ and $\mathcal{Z}_h$ is
\begin{align}
\frac{d\mathcal{Q}}{d\mathcal{Z}_h}(X) \propto \exp\left(-\int_{0}^{T} \phi_{\theta}(X_t)\,dt\right) \in [0,1] \quad \text{(since $\phi_\theta(x) \ge 0$)}. \label{eq:RN_der_BB}
\end{align}
The EA rejection samplers %
operate by proposing a path $X$ from   $\mathcal{Z}_h$ and accepting it with probability
$
\exp\left(-\int_{0}^{T} \phi_{\theta}(X_t)\,dt\right).
$
To do this, %
observe that~\cref{eq:RN_der_BB}  corresponds to the probability of zero events from a Poisson process with rate \(\lambda(t) = \phi_{\theta}(X_t)\) over \([0, T]\). %
Suppose for each $X$, we have an upper bound \(M_\theta(X) \geq \phi_{\theta}(X_t) \ \forall t\in [0,T]\). %
Then, by the Poisson thinning theorem~\citep{lewis1979simulation}, one can  simulate this Poisson process by first simulating from a  rate-\(M_\theta(X)\) homogeneous Poisson process and then keeping each event $e$ with probability $\phi_\theta(X_e)/M_\theta(X)$.
This only requires evaluating $X$ on the \emph{finite} set of times, together with the bound $M_\theta(X)$. %
If at the end, all Poisson events are thinned, we have a zero-event Poisson realization and $X$ is accepted.
If any Poisson events survive, the proposal $X$ is rejected, and we repeat with a fresh $X$. 
After acceptance, %
path values at any time can be simulated retrospectively from a Brownian bridge conditioned on $M_\theta(X)$ and its values on the Poisson grid.
Based on %
$M_\theta(X)$,  \citet{beskos2008factorisation}  define three families: %
\begin{enumerate}[(a)]
    \item\ \textbf{EA1-type diffusions}: Here \(\phi_\theta(\cdot)\) is uniformly bounded from above: $\phi_\theta(\cdot) \leq M_\theta$. %
    Now, with $M_\theta(X) = M_\theta$, the Poisson thinning procedure is straightforward: simulate from a rate-$M_\theta$ Poisson process, instantiate a biased Brownian bridge on this {finite} set of points, and keep each point $e$ with probability $\phi_\theta(X_e)/M_\theta$.

    \item\ \textbf{EA2-type diffusions}: Here, either \(\lim_{x \rightarrow \infty}\phi_\theta(x)\) or  \(\lim_{x \rightarrow -\infty}\phi_\theta(x)\) is finite. %
    If we assume
    the former, %
 $\phi_\theta(\cdot)$ is bounded over intervals of the form $[b,\infty)$. Given the path minimum $\Xmin =(\Xmin_v,\Xmin_t)$, set %
$M_\theta(X) = M_\theta(\Xmin)  \overset{\scriptscriptstyle\text{def}}{\ge} \sup_{x \in [\Xmin_v,\infty)} \phi_\theta(x)\ge \sup_{t\in[0,T]}\phi_\theta(X_t)$. %
Now, the Poisson-thinning scheme  first simulates the end-values $(X_0,X_T)$ from $h$, then the minimum $\Xmin$ of the Brownian bridge linking these values, before simulating a  rate-$M_\theta(\Xmin)$ Poisson process. 
Given $\Xmin$, imputing the Brownian bridge at the Poisson times requires simulating a Bessel bridge~\citep{makarov2009exact}. These are used to thin or keep the Poisson events as before.

\item\ \textbf{EA3-type diffusions}: Now $\phi_\theta(x)$ can tend to infinity as $x \rightarrow \infty$ and as $x \rightarrow -\infty$, and to bound $\phi_\theta(X_t)$ we need the minimum $\Xmin$ {\em and} maximum $\Xmax$ of $X$. In practice, we only bound these quantities.
Let $(\dotsc < L_2 < L_1 < U_1 < U_2 < \dotsc)$ be an increasing sequence, with $L_i \rightarrow -\infty$ and $U_i \rightarrow \infty$ as $i \rightarrow \infty$.  For a biased Brownian path $X$, define the {\em layer} $\layer$ %
as the smallest $i$ such that $X_t \in [L_i,U_i]\ \forall t \in [0,T]$. 
Associated with $\layer$ is the pair %
$(\Llayer,\Ulayer) := (L_{\layer},U_{\layer})$. %
Set $M_\theta(X) = M_\theta(\layer) \overset{\scriptscriptstyle\text{def}}{\ge} \sup_{x \in [\Llayer,\Ulayer]} \phi_\theta(x) \ge \sup_{t\in[0,T]}\phi_\theta(X_t)$. Then, simulate from a rate-$M_\theta(\layer)$ Poisson process, instantiate the Brownian bridge at these times  conditioned on $\layer$ and thin these as before. The second step is much more involved than before.
\end{enumerate}
Though not essential, with an additional lower bound $\phi_\theta(X_t) \ge m_\theta(X) \ge 0\ \forall t \in [0,T]$, the above schemes can be made slightly more efficient as described below: 
\begin{enumerate}[(a)]
    \item\ Initiate a proposal $X$ by simulating $(X_0,X_T)$ and $\Xmin$ (for EA2), or $\layer$ (for EA3). Use these to calculate $M_\theta(X)$ and $m_\theta(X)$. For EA1, these are uniform across $X$.
    \item\ With probability $1-\exp(-m_\theta(X)T)$ reject $X$ and go back to Step (a). Otherwise: 
    \item\ Simulate $\psi \subset [0,T]$ from a Poisson process with rate $\Delta_\theta(X) = M_\theta(X)-m_\theta(X)$, conditionally simulate $X_\psi$, and thin each $e \in \psi$ with probability $\frac{\phi_\theta(x_e)-m_\theta(X)}{\Delta_\theta(X)}$.
    \item\ If any event of $\psi$ is not thinned, reject $X$ and return to Step (a). Else accept $X$.
\end{enumerate}

\subsection{Challenges with the EA3 rejection sampler}
\label{sec:ea3_ch}
The EA3 sampler needs 3 operations, listed below in increasing order of complexity. %

\begin{enumerate}[(a)]
    \item\ {\em Simulate the layer $\layer$ of a path conditioned on its endpoints $X_0,X_T$:} %
    Let $\gamma_i$ be the probability that the $\layer$ equals $i$. %
    We can simulate $u \sim \text{Uniform}(0,1)$, and set $\layer=j$ if $u \in (\sum_{i \le j-1}\gamma_i, \sum_{i \le j}\gamma_i)$. While the $\gamma_i$'s are not available in closed form, \citet{pollock2016exact} construct converging and easy to evaluate upper and lower bounds to these. For any $u$, there exists a finite depth to which we need to evaluate these bounds to identify exactly which interval $u$ lies in. 
    To extend this scheme to simulate $\layer$ for a Brownian bridge conditioned on a finite set $G$, we simulate the layer of each segment $(g_i,g_{i+1})$ of the Brownian bridge, and %
    set the path layer $\layer$ to the largest of the layers of the individual segments\footnote{To make this possible, ours is a slightly different definition of layer from~\citet{pollock2016exact}}.
    \item\ {\em Simulate $p_{\cZ_h}(X_\psi \mid\layer)$, the values of the Brownian bridge at $\psi$ given $\layer$:} %
    In theory, for this step, we can propose the $|\psi|$-dimensional $\Xpr_\psi$ from a Brownian bridge, use the earlier step to simulate the layer ${\Xpr}^{\Updownarrow}$ of $\Xpr$, repeating these two steps until %
${\Xpr}^{\Updownarrow}$ equals $\layer$. In practice, %
\citet{pollock2016exact} use a cleverer proposal to improve the accept rate. Since our algorithm will not need this, we do not elaborate on this, and only emphasize how challenging and expensive this step is.
    \item\ {\em Simulate $p_{\cZ_h}(X_G \mid \layer,X_\psi)$, the values of the Brownian bridge at $G$ given its layer $\layer$ as well as values $X_\psi$ at the finite set of times $\psi$:} This step is the most complicated. %
    $\layer$ induces dependency across all intervals created by $\psi$, and one must simulate a number of \emph{local layers} from an exponentially large (in $|\psi|$) space. We then follow (b) to independently impute $X_G$ on each of the intervals given its sublayer. Section 8 of~\citet{pollock2016exact} has more details.
\end{enumerate}

\section{Theoretical framework for our MCMC approach}
\label{sec:theory_bg}
Consider a latent path $X \equiv \{X_t, t \in [0,T]\}$  modeled  (possibly after applying the Lamperti transform) as the SDE from~\cref{eq:sde}. We start by assuming %
the parameters $\theta$ are known, and drop $\theta$ from all subscripts. Given observations \( \vec{Y} = (y_1, y_2, \cdots, y_{|S|}) \) of $X$ at times \( S = (s_1 = 0 < s_2 < \cdots < s_{|S|} = T) \), %
write the likelihood of a path $X$ as
$l({Y} \mid X) = \prod_{i \in S} l(y_i \mid X_{s_i}) := l_Y(X_S)$.
The SDE forms a prior over paths with law $\cQ$, and given $Y$, we denote by \(\cQpost\) the law of the posterior. Write $\cZpost_h$ for the law of the posterior given ${Y}$ of the biased Brownian bridge $\cZ_h$. Bayes' theorem gives
$\frac{d\cZpost_{h}}{d\cZ_{h}}(X) \propto l_Y(X_S)$ and $ 
\frac{d\cQpost}{d\cQ}(X) \propto l_Y(X_S)$.
It follows from \cref{eq:RN_der_BB} that:
\begin{align}
    \frac{d\cQpost}{d\cZpost_h}(X) \propto \exp\left(-\int_{0}^{T} \phi(X_t)\,dt\right). \label{eq:post_dens}
\end{align}
Next, for any positive $M$, write  \(\mathcal{M}_M\) for the law of a  rate-$M$ Poisson process on \([0,T]\). We will use $\psi$ to denote Poisson process realizations. Recall that $\frac{d \mathcal{M}_M}{d \mathcal{M}_1}(\psi) \propto \exp(-MT)M^{|\psi|}$~\citep{reiss2012course}.
The results below generalize those from~\citet{wang2020exact} to EA2 and EA3 diffusions, and tighten them to include the minimum $m(X)$.
\begin{theorem}
Define a probability $\cP$ on the path-point process product space as
\begin{equation}
    \label{eq:gibbs_target}
\cP(dX,d\psi) \propto \cZpost_{h}(dX) \cdot \mathcal{M}_{\Delta(X)}(d\psi) \cdot \exp(-m(X)T) \prod_{g \in \psi} \left(\frac{M(X) - \phi(X_g)}{\Delta(X)}\right).
\end{equation}
 Then, the probability measure $\cQpost$ (i.e.\ the SDE posterior) is the marginal of $\cP$: %
$
\cQpost(dX) %
\propto \int_{\psi} \cZpost_{h}(dX) \times \exp(-m(X)T) \times \mathcal{M}_{\Delta(X)}(d\psi) \times \prod_{g \in \psi} \left( \frac{M(X) - \phi(X_g)}{\Delta(X)}\right).
$
\label{thrm:joint}
\end{theorem}
\noindent We prove this in the Supplementary Material; informally, this follows from the thinning theorem. The right hand side of~\cref{eq:post_dens} is the probability of zero events from a rate-$\phi(X_t)$ Poisson process, while~\cref{eq:gibbs_target} represents this as a thinned rate-$\Delta(X)$ Poisson process along with $0$ events from a rate-$m(X)$ Poisson process.
\Cref{thrm:joint} %
implies that for a pair \((X, \psi)\) simulated from~\cref{eq:gibbs_target}, discarding $\psi$ gives a path $X$ from $\cQ^+$.
While rejection sampling is one way to simulate from $\cP$ (for constant likelihood), writing it down allows other, potentially simpler approaches.
Towards this, we state two propositions that follow from~\Cref{thrm:joint}:
\begin{proposition}
\label{prop:condpath}
Write $\cP(\cdot \mid \psi)$ for the conditional law over paths under~\eqref{eq:gibbs_target}, for any $\psi$. Then this has density with respect to \(\mathcal{Z}_h\) given by:
\begin{align}
\frac{d \cP(X \mid \psi)}{d\mathcal{Z}_h} \propto 
l_Y(X_S) \cdot \exp(- M(X) \cdot T) \cdot \prod_{g \in \psi} \left({M(X)} - {\phi(X_g)}\right). \label{eq:path_cond}
\end{align}
\end{proposition}
\noindent The conditional $\cP(\cdot \mid X)$ over point process realizations %
has an even simpler form:
\begin{proposition}
\label{prop:condpoi}
Write $\cP(\cdot \mid X)$ for the conditional law of $\psi$ given $X$ under~\eqref{eq:gibbs_target}. %
For any $X$, this is the law of a Poisson process on \([0, T]\) with rate \(\lambda_X(t) = M(X) - \phi(X_t)\).
\end{proposition}

Slightly restricted versions of these two propositions form the theoretical basis of the EA1 MCMC sampler of~\citet{wang2020exact}, and will also guide the samplers we propose here. These samplers are all fundamentally Gibbs samplers that target the augmented distribution $\cP(dX,d\psi)$ from~\cref{eq:gibbs_target}, %
by alternately updating: 1) the path $X$ given the Poisson events $\psi$, and 2) the Poisson events $\psi$ given the path $X$. Below, we describe the sampler of~\citet{wang2020exact}, which we refer to as \texttt{GibbsEA1}.

\subsection{\texorpdfstring{The \texttt{GibbsEA1} sampler~\citep{wang2020exact}}{The GibbsEA1 sampler[wang2020exact]}} \label{sec:wang}
 For EA1 diffusions, $M(X)$,$m(X)$, and $\Delta(X)$ in~\cref{eq:gibbs_target} are constants $M$,$m$, and $\Delta$. Now, the propositions above can be directly implemented as two easy Gibbs steps that target $\cP(dX,d\psi)$.
The first updates the path $X$ given the Poisson events $\psi$. Proposition \ref{prop:condpath}, together with the Markov property of the Brownian bridge, implies we only need to update $X$ at times $\G$, with $X$ anywhere else conditionally following the Brownian bridge $\cZ_h(\cdot\mid X_\G)$. Furthermore, $X_{\G}$ has density with respect to Lebesgue measure given by $p(X_{\G}\mid\psi) \propto h(X_0,X_T)p_{\cZ_h}(X_{\G} \mid X_0,X_T)l_Y(X_S)\prod_{g\in\psi} (M-\phi(X_g))$.
Recalling from~\Cref{sec:notation} that $p_{\cZ_h}(X_{\G}\mid X_0,X_T)$ is a multivariate normal, $X_{\G}$ can be updated using standard MCMC tools for the Gaussian process literature (\citet{wang2020exact} use HMC).
The second Gibbs step updates $\psi\mid X$ from a Poisson process by Poisson thinning. %

Our work here extends this to EA2 and EA3 diffusions by additionally maintaining and updating, respectively, the minimum and layer information of the Brownian path. Below, we describe a na\"{\i}ve attempt to do this for EA3 diffusions:
\begin{description}
    \item[Update $\psi$ given $X$:]\ From Proposition~\ref{prop:condpoi},  $\psi \mid X$ follows an inhomogeneous Poisson process with rate $M(\layer) - \phi(X)$.
To simulate this, we use the thinning theorem, first simulating a homogeneous Poisson process $\vartheta$  with rate \( \Delta(\layer) \), and then keeping each point $e \in \vartheta$ with probability \( \frac{M(\layer) - \phi(X_e)}{\Delta(\layer)} \).

In practice, at the start of this step, the sampler has only instantiated $X_{\G}$ and $\layer$. From the latter, we calculate $\Delta(\layer)$ and simulate $\vartheta$. To thin this, we must impute $X_\vartheta$ given $\layer,X_\G$.
    Unfortunately, for EA3 diffusions, this involves the challenging Step (c) from~\Cref{sec:ea3_ch}, making each iteration extremely involved. Our actual MCMC algorithm avoids this. 
    We note that for EA2 diffusions, imputing $X_\vartheta \mid \Xmin,X_\G$ is an easier step involving Bessel bridge simulation.
    \item[Update $X$ given $\psi$:]\ This requires us to target the conditional $X_\G,\layer \mid \psi$. Unlike the EA1 situation, this is a much more challenging distribution because of $\layer$, and we cannot even write this up to a normalization constant. For EA2 diffusions (where we target $X_\G,\Xmin \mid \psi$), this expression is more tractable analytically; nevertheless, even here, this update is very challenging: the presence of Bessel process marginals makes MCMC updates like HMC complicated.
\end{description}

\section{\texorpdfstring{Our proposed \texttt{GibbsEA3} sampler for EA3 diffusions}{Our proposed GibbsEA3 sampler for EA3 diffusions}}
\label{sec:gibbs_ea3}
We describe our proposed MCMC sampler for EA3 diffusions, before explaining how it simplifies for the EA2 case. %
We introduce an independent auxiliary Poisson process $\xi$ with rate $\AuxRate$ over the interval $[0,T]$. Our MCMC state space is now $(X,\psi,\xi)$, with the target 
$
\cP(dX,d\psi,d\xi) := \cP(dX,d\psi)\cM_{\AuxRate}(d \xi)
$.  %
Here, $\cP(dX,d\psi)$ is the joint distribution from~\cref{eq:gibbs_target}; this, and therefore (from~\Cref{thrm:joint}) the SDE posterior $\cQ^+$  are marginals of $\cP(dX,d\psi,d\xi)$. %
We target this with a Gibbs sampler that alternately samples from the conditionals $\cP(d\psi,d\xi\mid X)$ and $\cP(dX,d\xi\mid \psi)$.

\subsection{\texorpdfstring{Conditionally updating $(\psi,\xi)$ from $\cP(d\psi,d\xi\mid X)$}{Conditionally updating (psi,xi) from P(d psi,d xi | X)}}
Our approach is to relabel the elements of $\psi \cup \xi$, producing a new pair $\psi'$ and $\xi'$, with $\psi \cup \xi = \psi' \cup \xi'$. This equality ensures no new values of $X$ given $\layer,X_{\G\cup\xi}$ need to be simulated, avoiding the difficulties associated with the $\psi$-update step of the na\"{\i}ve MCMC sampler. 
By construction, $\xi$ is %
independent of $X$ and $\psi$, so that $\cP(d\psi,d\xi \mid X) = \cP(d\psi \mid X){\cM_{\AuxRate}} (d\xi)$ while from
 Proposition~\ref{prop:condpoi}, $\cP(d\psi \mid X)$ is a rate-$M(\layer)-\phi(X_t)$ Poisson process. From the Poisson superposition theorem, $\psi \cup \xi$ is Poisson process with rate $\AuxRate+ M(\layer)-\phi(X_t)$. If we assign each element $e  \in \psi \cup \xi$ to $\psi'$ with probability $\frac{M(\layer)-\phi(X_t)}{\AuxRate+ M(\layer)-\phi(X_t)}$ and to $\xi'$ with probability  $\frac{\AuxRate}{\AuxRate+ M(\layer)-\phi(X_t)}$, then it follows from the Poisson thinning theorem that $
 \psi'$ and $\xi'$ are also independent Poisson processes with rates $M(\layer)-\phi(X_t)$ and $\AuxRate$ respectively.
Thus, the relabeling step forms a Markov kernel that has {$\cP(d\psi,d\xi\mid X)$} as its stationary distribution, and our proposed update forms a valid Metropolis-within-Gibbs step.

\subsection{\texorpdfstring{Conditionally updating $(X,\xi)$ from $\cP(dX,d\xi\mid \psi)$}{Conditionally updating X,xi from P(dX,dxi | psi)}}

From the independence of $\xi$ with $(X,\psi)$, we have $\cP(dX,d\xi\mid \psi) = \mathcal{M}_{\AuxRate}(d\xi) \mathcal{P}(dX \mid \psi)$. The last term $\mathcal{P}(dX \mid \psi)$ is given in Proposition~\ref{prop:condpath}, ~\cref{eq:path_cond}, and to further simplify this, note that we can decompose the probability of any path from the biased Brownian bridge $\cZ_h$ as (probability of its endpoints) $\times$ (probability of its values on $\G\cup\xi$ given end-point values) $\times$ (probability of its layer given earlier values)  $\times$ (probability of the rest of the path given everything so far).
Thus,
\begin{align}
\hspace*{-.12in}\cZ_h(dX) = p_{\cZ_h}(X_{\G\cup\xi}) dX_{\G\cup\xi} \player(\layer\mid X_{\G\cup\xi}) d\layer \cZ_h(dX\mid \layer,X_{\G\cup\xi}).
    \label{eq:bb_decomp}
\end{align}
This ordering differs from the EA3 rejection sampler, where, to simulate $\psi$, we must {\em first} simulate $\layer$, after which $X_{\G}$ follows a complicated distribution. Under our ordering, $X_{\G\cup\xi}$ is a simple Gaussian vector.  Of course, $\layer\mid X_{\G\cup\xi}$ still is a nonstandard distribution, however $\layer$ is much lower dimensional than $X_{\G\cup\xi}$ and simulating $\layer \mid X_{\G\cup\xi}$ is much easier than $ X_{\G\cup\xi} \mid \layer$. This is possible only because our Gibbs step \emph{conditions} on $\psi$.
Plugging~\cref{eq:bb_decomp} into~\cref{eq:path_cond}, we get%
\begin{align} \hspace{-.14in}
   \cP(dX,d\xi \!\mid \!\psi)  \propto \mathcal{M}_{\AuxRate}(d\xi)  p_Y(X_{\G\cup\xi},   \layer \!\mid\!  \psi)  dX_{\G\cup\xi} d\layer \cZ_h(dX\! \mid \! \layer, X_{\G\cup\xi})\ \hspace{-.1in}  \label{eq:tgt_decom}\\
\hspace{-.6in}\text{where }\    p_Y( X_{\G\cup\xi}, \layer \mid \psi)   \propto \bigg\{p_{\cZ_h}(X_{\G\cup \xi}) \cdot  l_Y(X_S) \bigg\}\cdot \player(\layer\mid X_{\G\cup \xi}) \cdot \nonumber \\ 
     \quad \ \    \bigg\{\exp(- M(\layer)\cdot T) \cdot \prod_{g \in \psi} \left(M(\layer) - {\phi(X_g)}\right)\bigg\} \nonumber\\
     \hspace{.8in}:= \bigg\{\piY(X_S) p_{\cZ_h}(X_{\psi\cup\xi} \mid X_S)\bigg\}  \cdot \player(\layer\mid X_{\G\cup\xi}) \cdot 
    \bigg\{\piM({\layer}, X_\psi)\bigg\}.\nonumber
\end{align}
Note that the term $\piY(X_S)= p_{\cZ_h}(X_S)l_Y(X_S)$ is a standard Gaussian process (GP) posterior.
To target \(\cP(dX,d\xi\mid\psi)\), we employ a Metropolis-Hastings (MH) scheme, proposing  from some probability \(\mathcal{V}(d\Xpr,d\xipr\mid X, \xi, \psi)\), and accepting with the usual MH acceptance probability. Along the lines of~\cref{eq:tgt_decom},  we decompose $\mathcal{V}$ as follows:
\begin{align*}
    \mathcal{V}(d\Xpr,d\xipr \mid X,\xi,\psi) &= \cM_{\AuxRate}(d\xipr)q(\Xpr_{S} \mid X_{S})p_{\cZ_h}(\Xpr_{\psi\cup \xipr}\mid\Xpr_S)d\Xpr_{\G\cup\xipr}\\
    &\quad\quad p^{\Updownarrow}_{\cZ_{h}}(\Xpr^{ \Updownarrow} \mid \Xpr_{\G\cup\xipr}) d\Xpr^{\Updownarrow}\cZ_h(d\Xpr \mid \Xpr^{ \Updownarrow},\Xpr_{\G\cup\xipr}).
\end{align*}
That is, we first draw $\xipr$ from a rate-$\AuxRate$ Poisson process, propose new values $\Xpr_{S}$ conditioned on the old values $X_{S}$. Conditioned on these (and independent of previous MCMC values), we propose $\Xpr_{\psi\cup\xipr}$ from a Gaussian Brownian bridge marginal, and then propose a new layer ${\Xpr}^{\Updownarrow}$, with the remaining path following the intractable Brownian conditional. 
This yields the following acceptance probability\\
    $\hspace*{1in} \text{acc} = \frac{\cP(d\Xpr,d\xipr \mid \psi)}{\cP(dX,d\xi \mid \psi)}\cdot\frac{d\mathcal{V}(X,\xi \mid \Xpr,\xipr,\psi)}{d\mathcal{V}(\Xpr,\xipr \mid X,\xi,\psi)}
    = \frac{\piY(\Xpr_{S}){\piM(\Xpr^{\Updownarrow},\Xpr_\psi)}}{\piY(X_{S}){\piM(\layer,X_\psi)}}\cdot\frac{q(X_{S} \mid \Xpr_{S})}{q(\Xpr_{S} \mid X_{S})}.$\\
Notice that the intractable $\player(\cdot)$ terms cancel out above, %
and we are only left to find an efficient proposal $q( \Xpr_{S}\mid  X_{S} )$.
{ In general, $\piY(X_{S})$ is a standard GP posterior  distribution, and we can employ any Markov kernel that targets it. From detailed balance, such a kernel satisfies $q( \Xpr_{S} \mid X_{S} ) \piY(X_{S}) = q(X_{S} \mid \Xpr_{S} ) \piY(\Xpr_{S})$, so that the %
acceptance probability further simplifies to
\begin{align}
    \text{acc} = 
     \frac{{\piM(\Xpr^{\Updownarrow},\Xpr_\psi)}}{{\piM(\layer,X_\psi)}}= \frac{\exp(-M(\Xpr^\Updownarrow)\cdot T)}{\exp(-M(\layer)\cdot T)} \cdot \frac{\prod_{e\in \psi}\left(M(\Xpr^{\Updownarrow})-\phi(\Xpr_{e})\right)}{\prod_{e\in \psi}\left(M(\layer)-\phi(X_{e})\right)}. \label{eq:mh_ea3}
\end{align}
For a non-Gaussian likelihood, we set $q$ to a Hamiltonian Monte Carlo (HMC) kernel: our target is the posterior distribution for a Gaussian prior with a non-conjugate (but differentiable) likelihood, a setting where HMC is known to excel. Here, one iteration of HMC involves $L$ leapfrog steps of size $\epsilon$ (with $L, \epsilon$ and a mass-matrix $M$ parameters of the algorithm) followed by an accept/reject step~\citep{neal2011mcmc}. For a Gaussian likelihood $l_Y(X_S)$, $\piY$ admits a conditionally Gaussian structure where the distribution of $X_S$ given $(X_0,X_T,Y)$ is Gaussian. Thus we can run HMC on the lower-dimensional target $h(X_0,X_T)p(Y \mid X_0, X_T)$ to sample $(\widetilde{X}_0,\widetilde{X}_T)$, and then draw $\widetilde{X}_S\mid \widetilde{X}_0,\widetilde{X}_T$ from the Gaussian $p_{\cZ}(X_S\mid X_0, X_T)l_Y(X_S)$. For an endpoint-conditioned diffusion we do not even need the HMC step.}

Thus in general, our scheme has two accept/reject steps each iteration: one for the HMC proposal (to correct for discretizing the Hamiltonian dynamics), and one to correct the mismatch between the overall proposal and target distributions. %

\begin{algorithm}[h]
   \caption{{One iteration of \texttt{GibbsEA3} for EA3 diffusions.}}
   \label{alg:ea3gibbsparainf}
   \begin{flushleft}
   \textbf{Input:}
   \begin{itemize}[leftmargin=1.5em, labelsep=0.5em, itemsep=0pt, topsep=0pt, parsep=0pt, partopsep=0pt]
      \item Observation times $S \subset [0,T]$, along with likelihood function $l_Y(X_S)$. %
      \item Poisson times $\psi \subset [0,T]$, path values $X_{\G}$ on $\G$, and layer $\layer$.
      \item Auxiliary Poisson times $\xi \subset [0,T]$ and path values $X_{\xi}$.
   \end{itemize}
   \textbf{Output:}  
   \begin{itemize}[leftmargin=1.5em, labelsep=0.5em, itemsep=0pt, topsep=0pt, parsep=0pt, partopsep=0pt]
      \item Updated Poisson times $\psi'$, $\xi'$, path values $X'_{\psi' \cup S\cup \xi'}$, and layer ${\layerpr}$. %
   \end{itemize}
   \vspace{0.3em}
   \hrule
   \end{flushleft}
   \vspace{-0.5em}
   \begin{algorithmic}[1]
    \State \textbf{Update $(\psi,\xi)$ conditioned on $X$:}
      \State \quad Assign each $e \in \psi \cup \xi$ to $\psi'$ with probability $\frac{M(\layer)-\phi(X_t)}{\AuxRate+M(\layer)-\phi(X_t)}$, else to $\xi'$. %
    \State \textbf{Update $(X,\xi)$ conditioned on $\psi$:}
      \State %
\quad      Simulate $\Xpr_S \mid X_S$ from a Markov kernel  targeting $p_{\cZ_h}(X_{S})l_Y(X_S)$ (e.g.\ HMC). %
      \State \quad Simulate a rate-$\AuxRate$ homogeneous Poisson process $\xipr$ on $[0,T]$.
      \State \quad Simulate $\Xpr_{\psi' \cup \xipr} \mid \Xpr_S$ from a Brownian bridge.
      \State \quad Simulate the layer $\Xpr^{\Updownarrow}$ of a Brownian bridge given $\Xpr_{\psi'\cup S\cup \xipr}$.
      \State \quad  Set $(\xi',X'_{\psi'\cup S\cup \xi'}, {\layerpr})$ to $(\xipr,\Xpr_{\psi'\cup S\cup \xipr}, {\Xpr^{\Updownarrow}})$ with probability in~\cref{eq:mh_ea3}.
\State \textbf{Return} the updated values $(\psi',\xi',X'_{\psi'\cup S\cup \xi'}, {\layerpr})$.
   \end{algorithmic}
\end{algorithm}

\subsection{Comments on the algorithm} \label{sec:comm}
Our algorithm above can be simplified slightly for EA2 case, where we include $\Xmin$ instead of $\layer$ in the MCMC state space. Now, we no longer need to include $\xi$, and can directly update $\psi \mid X$ using the simpler scheme outlined at the end of~\Cref{sec:wang}: given $\Xmin$, we can easily simulate $X$ anywhere else from a Bessel process.
Our scheme to update $X \mid \psi$ is needed for both EA2 and EA3 diffusions. With the HMC update of $X_S$, this is a double Metropolis-Hastings scheme~\citep{liang2010double}, having %
two MH accept-reject steps. 
The first rejection is a concern only for non-Gaussian likelihoods, and can be controlled with smaller (but more) HMC leapfrog steps. We do not recommend investing too much computation towards this, since this can eventually still be rejected. One can also use other, cheaper schemes (often easier to tune than HMC) that target $p_{\cZ_h}(X_S)l_Y(X_S)$ such as elliptical slice sampling~\citep{murray2010elliptical}. Even in the event of a rejection, we continue with the proposal, now with $\Xpr_S$ equal to the current state $X_S$, but with new values  $\Xpr_{\psi\cup \xipr}$ and $\Xpr^{\Updownarrow}$.

The second accept/reject event  %
(\cref{eq:mh_ea3})
reflects how compatible $\Xpr_\psi$ and $M(\Xpr^{\Updownarrow})$ are with the Poisson events $\psi$. 
For instance, the term $\prod_{e \in \psi} \left({M({\Xpr^\Updownarrow})} - {\phi(\Xpr_e)}\right)$ suggests we do not want too many components $e \in \psi$ with $\phi(\Xpr_e)$ large. 
In our experiments, we did not encounter very low acceptance rates. 
Still, we provide some guidelines for such situations. An obvious approach follows~\citet{beskos2006exact, sermaidis2013markov}, and divides $[0,T]$ into $B$ subintervals $[0,T/B],[T/B,2T/B],\dotsc,[(B-1)T/B,T]$. We can now split our path update step into $B$ steps, with the $b$th updating 
$X_{[\frac{b-1}{B}T, \frac{b}{B}T]}$ given $\psi,X_\psi$ and the rest of the path. Unlike the earlier works, where the subintervals must align with the observation times $S$, we are free to choose these as we please. Two more ideas are:

\noindent\textbf{Avoiding bad $\Xpr_\psi$ proposals during the path update step:} A simple idea is to not update $X_\psi$ at all this step, instead, following a slightly different Gibbs scheme: simulate $\xipr, \Xpr_{S \cup \xipr}\mid\psi,X_\psi$ (so that $\tilde{X}_\psi = X_\psi$), and then update $\psi,\xipr \mid X_\psi,X_{\xipr}$. This is still ergodic, since every time an element $e \in \psi$ is swapped into $\xipr$, that $X_e$ can be updated.  %
This is even more efficient for the case of EA2 diffusions, where we directly update $\psi,X_\psi$ without introducing $\xi$.
From the Gaussianity and Markovianity of Brownian motion, we can still easily target $p_{\cZ_h}(X_S \mid X_\psi)l_Y(X_S)$. 

\noindent\textbf{Improving $\Xpr^{\Updownarrow}$ proposals during the path update step:} With the earlier change, we only have to maintain reasonable accept rates when updating the $\layer$.
A simple approach is to propose multiple $\layer$'s (this can be efficiently vectorized), and use a multiple-try MH accept probability~\citep{liu2000multiple}.

{\section{Parameter inference}\label{sec:bayesframework}

A natural way to extend our sampler to also simulate the SDE parameter(s) $\theta$ follows~\citet{beskos2006exact, sermaidis2013markov, wang2020exact} and constructs an expanded Gibbs sampler that alternately a) simulates a path skeleton (specifically $\psi, X_\G, \Xmin$ for EA2 and $\psi,\xi,X_{\G \cup \xi},\layer$ for EA3) given $\theta$, exactly following the earlier methodology, and b) the parameter $\theta$ given the path skeleton.
The result below extends~\Cref{thrm:joint} to make all dependency on $\theta$ explicit:
\begin{proposition}
\label{prop:parjoint}
With a prior $\ptheta(\theta)$ on $\theta$, the joint law of $\theta, X, \psi$ and $Y$ is given by %
\begin{align}
    \mathcal{P}(d\theta,dX,d\psi) 
        \propto & \ptheta(d\theta)l_Y(X)\cZ_{h_{\theta}}(dX)\mathcal{M}_{\Delta_\theta(X)}(d\psi)\exp(-m_\theta(X)\cdot T) \nonumber \\
        & \prod_{g\in \psi}\left(\frac{M_\theta(X)-\phi_\theta(X_{g})}{\Delta_\theta(X)}\right)\times
        {\exp\left\{-\alpha_\theta^{\downarrow}T+\log c_{\theta}\right\}}. %
\end{align}
\end{proposition}
As $\cZ^+_{h_\theta}(dX) \propto h_\theta(X_0,X_T)dX_0dX_T \cZ(dX | X_0,X_T)l_Y(X)$, the conditional $\theta \mid X,\psi$ is then proportional to the above expression, with the likelihood term $l_Y(X)$ dropped.
This distribution, while nonstandard, can be simulated from using standard %
MCMC techniques such as MH or slice sampling. A limitation is the possibility of strong coupling between $\theta$ and $X$, slowing convergence. \citet{beskos2006exact, sermaidis2013markov} describe non-centered reparametrizations that can alleviate this.
We outline how to extend our \texttt{GibbsEA3} algorithm to \emph{jointly} update $X$ and $\theta$ while conditioning only on $\psi$. The same idea also applies to \texttt{GibbsEA2}.

The first step, updating $\psi,\xi \mid \theta,X$, remains the same. For the second step, 
we now target the joint $\mathcal{P}(d\theta,dX,d\xi \mid \psi) = \mathcal{P}(d\theta,dX \mid \psi)\mathcal{M}_{\AuxRate}(d\xi)$. Accordingly, we modify the first part of this step (that updates $X_S$ to $\widetilde{X}_S$) to now jointly update $\widetilde{X}_S$ and $\widetilde{\theta}$ using a Markov kernel targeting $\ptheta(\theta)p_{\cZ_{h_\theta}}(X_S)l_Y(X_S) := \ptheta(\theta)h_{\theta}(X_0,X_T)p_{\cZ}(X_S\mid X_0, X_T)l_Y(X_S)$. Here, we drop the $h$-subscript in $p_{\cZ}(X_S\mid X_0, X_T)$ (recall this is just the Gaussian marginal of a Brownian bridge) to emphasize it does not depend on $\theta$. For this step, we can either continue to use HMC (with state-space augmented to include $\theta$), or use one of the many MCMC schemes for posterior simulation of Gaussian process hyperparameters~\citep[e.g.][]{murray2010slice}. For the common setting of Gaussian noise, this is especially simple since $p(Y \mid X_0, X_T)$ is tractable: run HMC targeting $\ptheta(\theta)h_\theta(X_0,X_T)p(Y \mid X_0, X_T)$ to produce $(\widetilde{\theta},\widetilde{X}_0,\widetilde{X}_T)$, and then simulate $\widetilde{X}_S\mid \widetilde{X}_0,\widetilde{X}_T$ from the Gaussian $p_{\cZ}(X_S\mid X_0, X_T)l_Y(X_S)$. Whichever kernel we choose, it is easy to see that the acceptance probability is %
\begin{align}
    \text{acc} 
    &=     \frac{ \exp(- M_{\widetilde{\theta}}({\Xpr^{\Updownarrow}}) \cdot T) \cdot 
     \prod_{e \in \psi} \left({M_{\widetilde{\theta}}({\Xpr^{\Updownarrow} })} - {\phi_{\widetilde{\theta}}(\widetilde{X}_e)}\right)\cdot \exp\{-\alpha^{\downarrow}_{\widetilde{\theta}}T + \log c_{\widetilde{\theta}}\}} %
     { \exp(- M_\theta(\layer) \cdot T) \cdot 
     \prod_{e \in \psi} \left({M_\theta(\layer)} - {\phi_\theta(X_e)}\right)\cdot \exp\{-\alpha^{\downarrow}_{\theta}T + \log c_{\theta}\}}. %
     \label{eq:path_par_acc}
\end{align}
Our algorithm now has the flavor of~\citet{zhang2021mjp}, but adapted from the setting of Markov jump processes to our more complex setting.

\section{Experiments}
In this section and in the Supplementary Material, we  assess the performance of our MCMC samplers on several examples within the classes EA2 and EA3.
Unless specified otherwise, we ran our Gibbs samplers for 10,000 iterations, using the first 2,000 as burn-in. We consider three versions of our proposed algorithms, depending on the choice of the Markov kernel targeting $\pi_Y(X_S) = p_{\cZ_h}(X_S)l_Y(X_S)$ in~\cref{eq:tgt_decom}: 1) a general HMC-based approach we call~\texttt{GibbsHMC}, 2) another general (but inefficient) independent Metropolis-Hastings kernel, where we target $\pi_Y(X_S)$ with proposals from the prior $p_{\cZ_h}(X_S)$ (we call this \gibbsprior), and 3) one applicable only to Gaussian likelihoods,  {where we sample from $\pi_Y(X_S)$ using a combination of lower-dimensional HMC and exact Gaussian sampling} (we call this \gibbspost). 
For the HMC kernel of \gibbshmc, we used fairly default hyperparameter values of \(\epsilon = 0.1, L = 10\), and $|S|\times|S|$ mass matrix $M$ set to the Hessian of the log-posterior distribution at its mode. We include results in the supplement demonstrating the robustness of our algorithm to these settings.
Recall  our HMC sampler operates on a fixed $|S|$-dimensional space; this is in contrast to~\citet{wang2020exact} which additionally operates on $X_\Psi$. Since the cardinality of $\Psi$ changes from iteration to iteration, tuning $M$ is much simpler in our case.

For \texttt{GibbsEA3}, we set the rate of auxiliary Poisson process to $\AuxRate = 2$. %
We would expect that increasing this improves mixing at the cost of greater computation, though in the supplement we show how this has only a small effect on overall mixing.

We also include as a baseline the
particle MCMC sampler~\citep{andrieu2010particle} (\pmcmc): this is a flexible and standard algorithm for posterior path and parameter simulation (whose simplicity in our setting comes at the cost of discretization bias). %
Because of its greater computational expense, we ran \pmcmc\ for 5,000 iterations with a discretization timestep \( \Delta t = 0.001 \) and with 50 particles. In~\Cref{sec:parinf_cir}, we also compare with \pmcmc\ with no discretization error.

We quantify sampler efficiency using effective sample size per unit time (ESS/s) obtained by dividing the median effective sample sizes for a number of statistics (in most cases, the path values at each observation, and if applicable, also the parameter) by the wall-clock runtime. %

\subsection{\texorpdfstring{EA2 example: $dX_t = p \cdot \exp(-qX_t) \, dt + dW_t$}{EA2 example}}
\label{SDE3}
We consider the above SDE, with parameters $\theta = (p,q)$ where $p > 0$ and $q > 0$; we show in the supplement that this belongs to class EA2.

\noindent \textbf{Gaussian Observations:}
With $p=q=1$, and $X_0=-1$, we simulate paths from this SDE over the interval $[0,T]$.
We then generate $|S|$ observations at times $S=(s_1,\dotsc,s_{|S|})$ spread uniformly over $[0,T]$, with each observation Gaussian distributed, centered at the path value at its time, and with variance $\sigma^2_N$. %
We generate 10 datasets for each configuration of $|S|$ in the set $\{1, 5, 10, 20\}$, $T$ in $\{1,2,5,10\}$, and $\sigma_N \in \{0.1, 1\}$. %
\Cref{fig:ex1_ess_over_time_ln} presents boxplots of effective sample sizes per unit time for different samplers across the different settings.
 
\begin{figure}
 \centering
\includegraphics[width=0.8\textwidth]{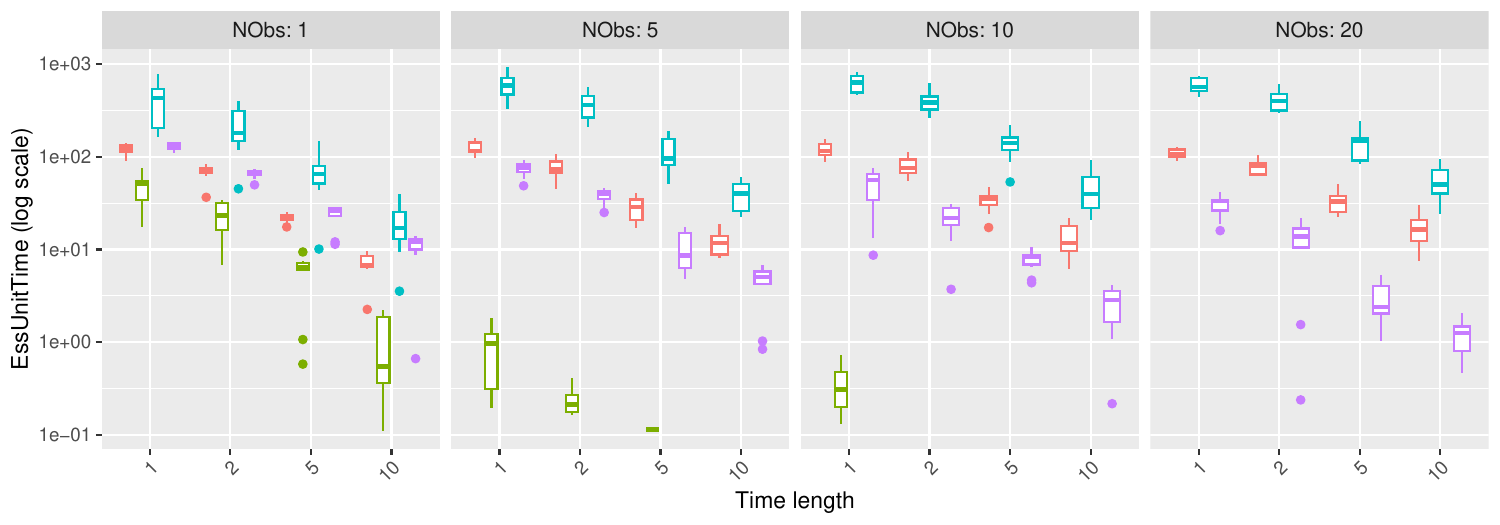}
 \centering
 \includegraphics[width=0.8\textwidth]{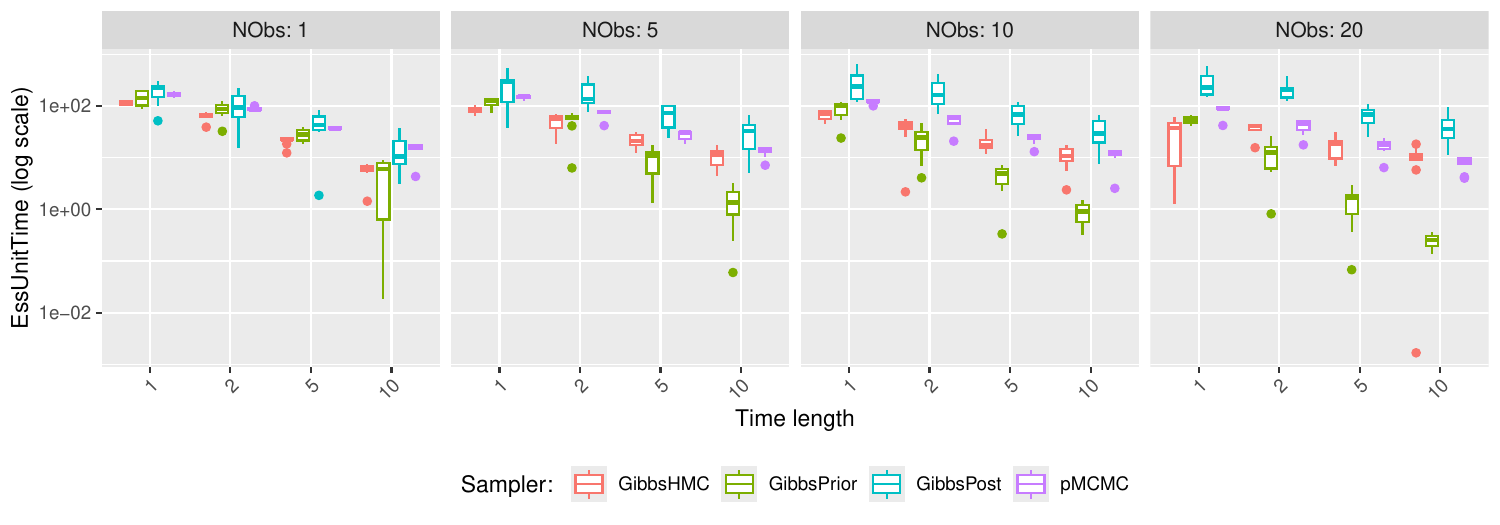}
  \caption{The effective samples per second vs time lengths of the EA2 SDE under Gaussian noise $\sigma_N = 0.1$ (top row) and $\sigma_N=1$ (bottom row)}
\label{fig:ex1_ess_over_time_ln}
\end{figure}

In this setting with Gaussian noise, \gibbspost\ performs best, as we expect. On the other hand, the simplest proposal scheme \gibbsprior\ is the worst performer, especially in settings with small variance \(\sigma_N = 0.1\), a large number of observations, and extended time lengths. %
For clarity, we exclude the worst performances of \gibbsprior\ from the plots.
\gibbshmc\ sits between these two performance-wise, and despite being applicable across a much larger class of likelihoods, exhibits a performance that is competitive with the best sampler. 
It is never less effective than the biased \pmcmc\ sampler in any setting. At the same time, as the sampling problem becomes harder (e.g.\ longer observation time, more observations and less noisy observations), it starts to outperform \pmcmc, with %
the disparity more pronounced in scenarios with smaller Gaussian noise variance. %

\subsection{\texorpdfstring{EA3 example (double-well potential Langevin process): $   dX_{t} = (-pX_{t}^{3}+qX_{t})dt + dW_{t}$}{EA3 example (double-well potential Langevin process)}}
\label{sec:EA3_example}
In the supplement, we verify the above SDE is of class EA3, and derive associated quantities. We set the parameters $p = 1/8$ and $q = 1/2$, and with $X_{0} = 0$.

\noindent \textbf{Gaussian and Poisson Observations:}
With the same setup as earlier, \Cref{fig:gnex6_ess_over_time_ln} compares the performance of \gibbshmc, \gibbspost\ and \pmcmc. Despite our sampler for the EA3 setting being slightly more complex than the EA2 setting, the story is the same as before; if anything, the gap between our proposed \gibbshmc\ and \pmcmc\ is even more clear.
\begin{figure}
 \centering
 \includegraphics[width=0.8\textwidth]{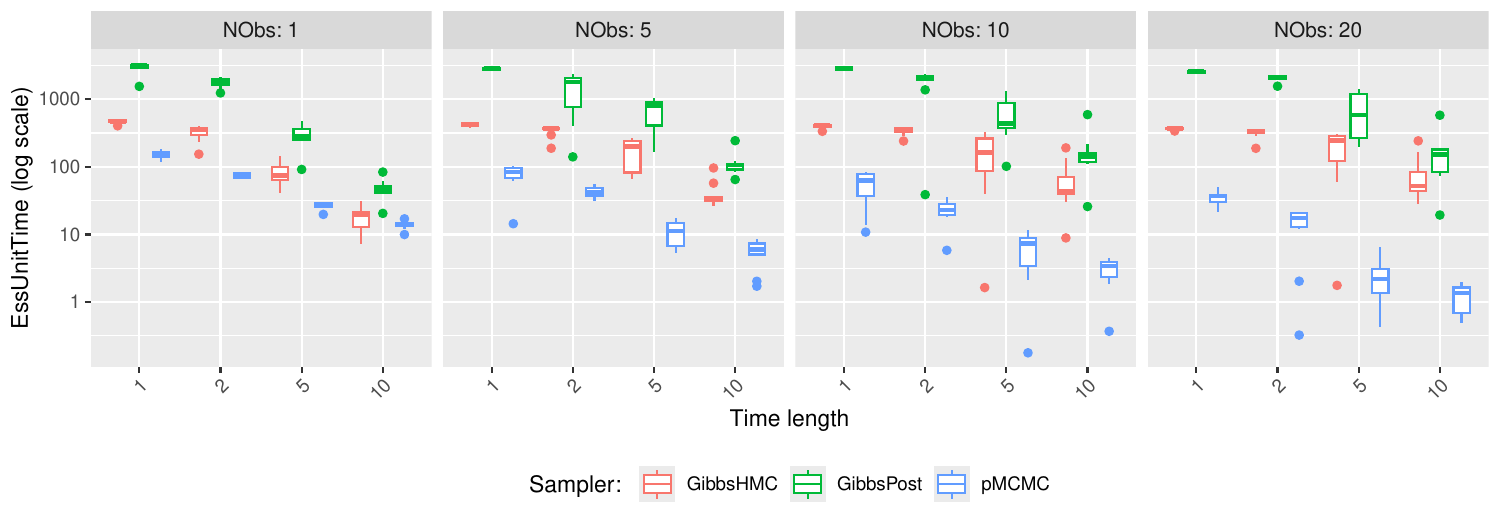}
 \includegraphics[width=0.8\textwidth]{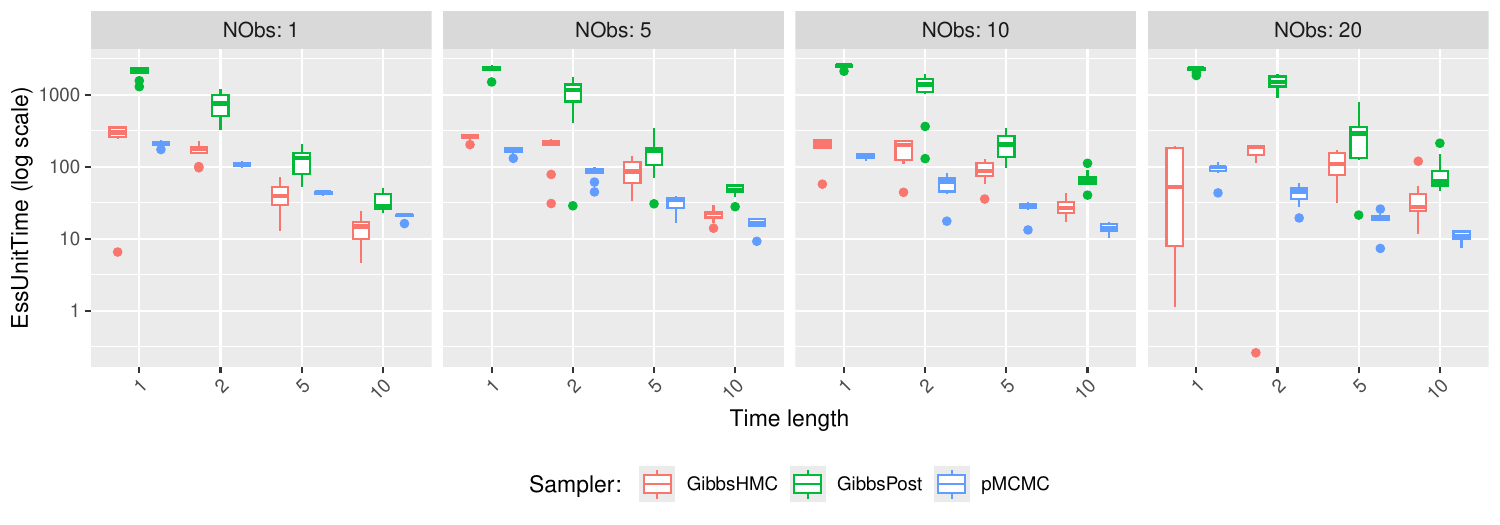}
  \includegraphics[width=0.8\textwidth]{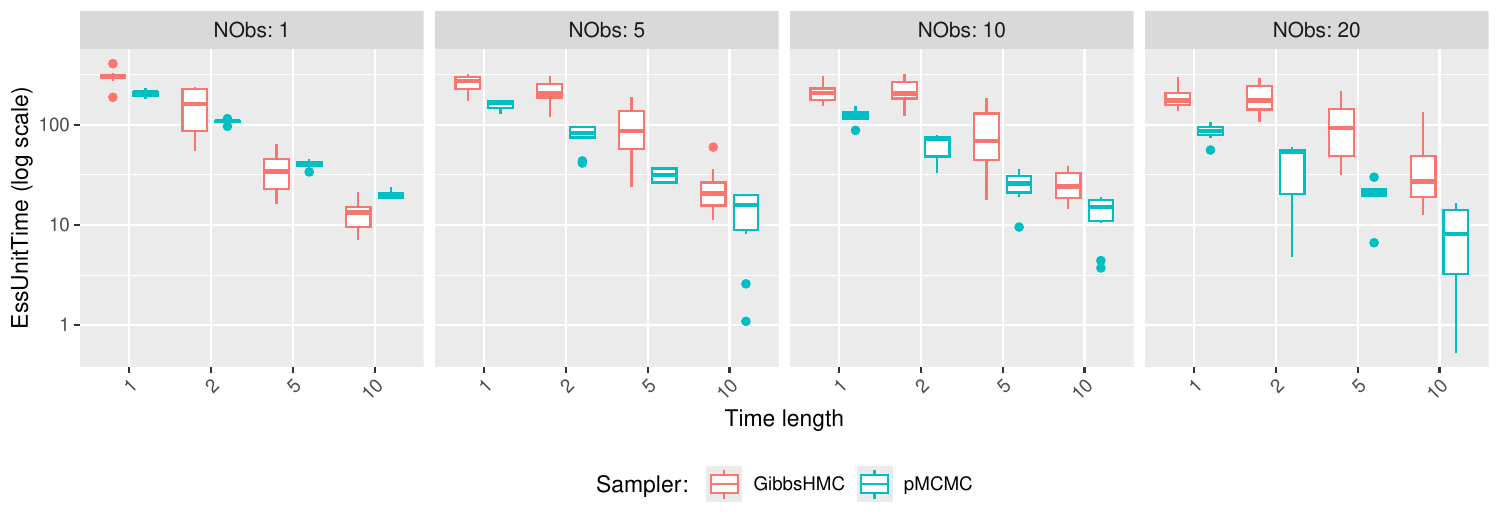}
  \caption{Effective samples per second vs time lengths of EA3 double-well diffusion under Gaussian noise $\sigma_N = 0.1$ (top), $\sigma_N=1$ (middle) and Poisson noise (bottom)}
  \label{fig:gnex6_ess_over_time_ln}
\end{figure}
\Cref{fig:gnex6_ess_over_time_ln} also displays boxplots of effective sample sizes per unit time for the setting with Poisson likelihoods: now each observation is a non-negative integer, with $y_{i} \sim \text{Poi}(\exp(X_{s_{i}}))$.
Now, \gibbspost\ is no longer applicable, but otherwise, the results here show the same pattern as before: our exact approach is competitive with the approximate \pmcmc, but quickly dominates it as the effect of observations becomes stronger and stronger.

\subsection{Parameter inference: CIR process}
\label{sec:parinf_cir}
We illustrate our joint path-parameter update scheme  for the  Cox--Ingersoll--Ross Process (CIR) \citep{cir1985} $dV_t = p(q - V_t)dt + \sigma\sqrt{V_t}dW_t$. Here, $\theta = (p,q,\sigma)$ with $p > 0$, $q > 0$, and $\sigma > 0$. %
Applying the transform $X_t = 2\sqrt{V_t}/\sigma$ yields a unit diffusion coefficient process 
 $ \hspace*{.1in}   dX_t = %
    \left\{\frac{1}{X_t}\left(\frac{2pq}{\sigma^2}-\frac{1}{2}\right)-\frac{pX_t}{2}\right\}%
    dt + dW_t.
$ 
Define the degree $d = 4pq/\sigma^2$ and assume that $d \geq 3$. Under this condition, %
the diffusion belongs to the EA3 class \citep{reutenauer2008} (see also the supplement), %
with
$\phi_{\theta}(x) = \left\{\left(\frac{2pq}{\sigma^2}-1\right)^2 - \frac{1}{4}\right\}\frac{1}{2x^2} + \frac{p^2}{8}x^2 - \frac{p}{4}\left\{\sqrt{(d-1)(d-3)}\right\}$.

We simulate a CIR path with $X_0 = 3.5$  and parameters $(p,q,\sigma) = (1.6,1.1,0.6)$ over the interval $[0,10]$. We record path values at 250 equally spaced observations, each corrupted by independent Gaussian noise with standard deviation $0.2$. 
 We place independent rate-1 exponential priors on $\theta = (p,q,\sigma)^{\top}$, and assume $X_0$ is known.

We implement the joint path-parameter update scheme from~\Cref{sec:bayesframework}.
Since the observation model is Gaussian, we employ the \gibbspost\ version of \texttt{GibbsEA3}. %
\Cref{eq:path_par_acc} requires evaluation of the normalizing constant of $h_{\theta}(X_T \mid X_0)$, $c_{\theta,X_0}$, as well as its gradient $\nabla_{\theta} c_{\theta,X_0}$. These are both straightforward one-dimensional integrals, and are computed numerically using the \texttt{R} function \texttt{integrate()}. %
We compare with \pmcmc, now incorporated into a particle Metropolis-Hastings scheme, with parameters proposed from a random-walk Gaussian proposal with covariance $\Sigma = 0.1\times \mathbb{I}_3$. We note that the CIR process admits a closed-form non-central chi-squared transition density, so that our methodology is not strictly necessary. Nevertheless, this allows us to implement and compare with \pmcmc\ {\em without} time-discretization error. This algorithm is more efficient than the time-discretized version, since the number of steps equals the number of observations, rather than the cardinality of the time-discretization grid. %

Table~\ref{table:pi_gncir} summarizes posterior statistics for $(p,q,\sigma)$ obtained from both samplers, across ten independent runs. %
Both algorithms yield comparable posterior summaries. Across all parameters and repetitions, \gibbspost\ consistently attains markedly higher ESS per second than \pmcmc; this is despite the fact that \pmcmc\ (unlike our method) exploits the tractable transition probabilities of the CIR process.
\begin{table}[!h]
\scriptsize
\centering
\begin{tabular}{ccccccc}
\toprule
Algorithm & Param & Mean & Median & Mode & SD & ESS (per second) \\ \midrule %
\texttt{GibbsPost} & $p$ & $1.220$ & $1.217$ & $1.330$ & $0.522$ & $0.799$\\
 & $q$ & $1.202$ & $0.935$ & $0.471$ & $1.027$ & $6.156$ \\
 & $\sigma$ & $0.565$ & $0.536$ & $0.492$ & $0.254$ & $4.939$\\
\midrule
\texttt{pMCMC} & $p$ & $1.131$ & $1.090$ & $1.148$ & $0.497$ & $0.082$\\
 & $q$ & $1.256$ & $0.925$ & $0.457$ & $0.990$ & $0.194$ \\
 & $\sigma$ & $0.552$ & $0.524$ & $0.531$ & $0.258$ & $0.136$\\
\bottomrule
\end{tabular}
\caption{Comparison of MCMC summary statistics for the CIR process } %
\label{table:pi_gncir}
\end{table}

\subsection{Real-world example: Application to Ice-core data}
\label{sec:real_example}

We consider a dataset\footnote{available at {\tiny\url{https://www.ncei.noaa.gov/pub/data/paleo/icecore/greenland/summit/ngrip/isotopes/ngrip-d18o-50yr-noaa.txt}}} of measurements of the oxygen isotopic composition of the ice ($\delta^{18}\text{O}$) in Greenland. Such measurements were analyzed in~\cite{Kwasniok2009icecore,Kwasniok2012icecore,Kwasniok2013icecore}, and serve as a proxy for Northern Hemisphere temperature variations \citep{NGRIP2004Nature}. We focus on the period spanning $60,000$ up to $20,000$ years before the present. We subsampled the original time series at $250$-year intervals, yielding a total of $160$ observations. We centered the data, and %
rescaled by $\sigma = 3.9$ (the estimate of the diffusion coefficient reported in \cite{Kwasniok2013icecore}), in order to conform to a unit-diffusion coefficient formulation. Finally, we rescaled the observation interval to $[0,4]$. The raw and pre-processed data are displayed in the left and center panels of~\Cref{fig:app_icecore_data}.

\begin{figure}[tbh!]
    \centering
    \begin{minipage}{0.29\textwidth}
        \centering
        \includegraphics[width=\linewidth]{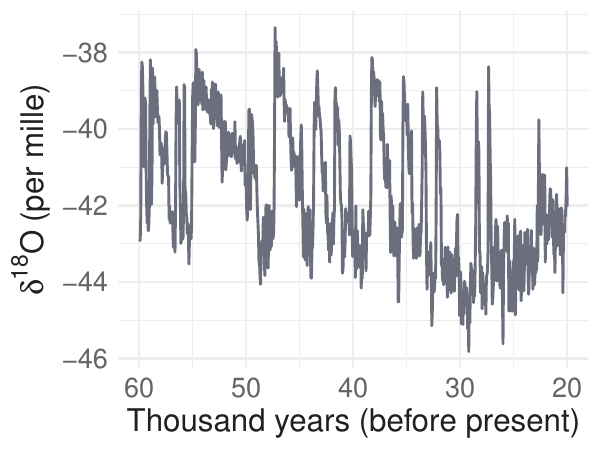}
    \end{minipage}
    \begin{minipage}{0.29\textwidth}
        \centering
        \includegraphics[width=\linewidth]{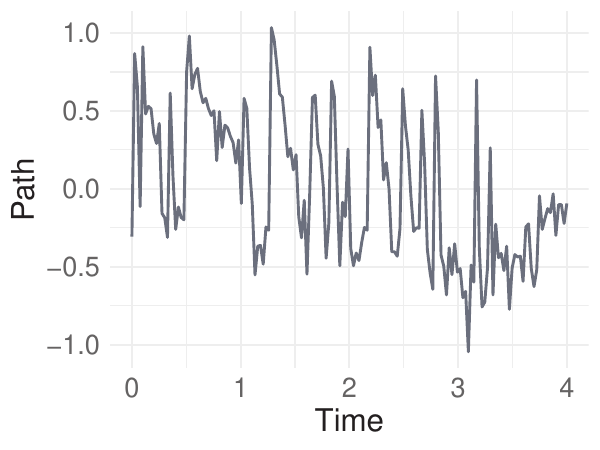}
    \end{minipage}
        \begin{minipage}{0.4\textwidth}
        \centering
        \includegraphics[width=\linewidth]{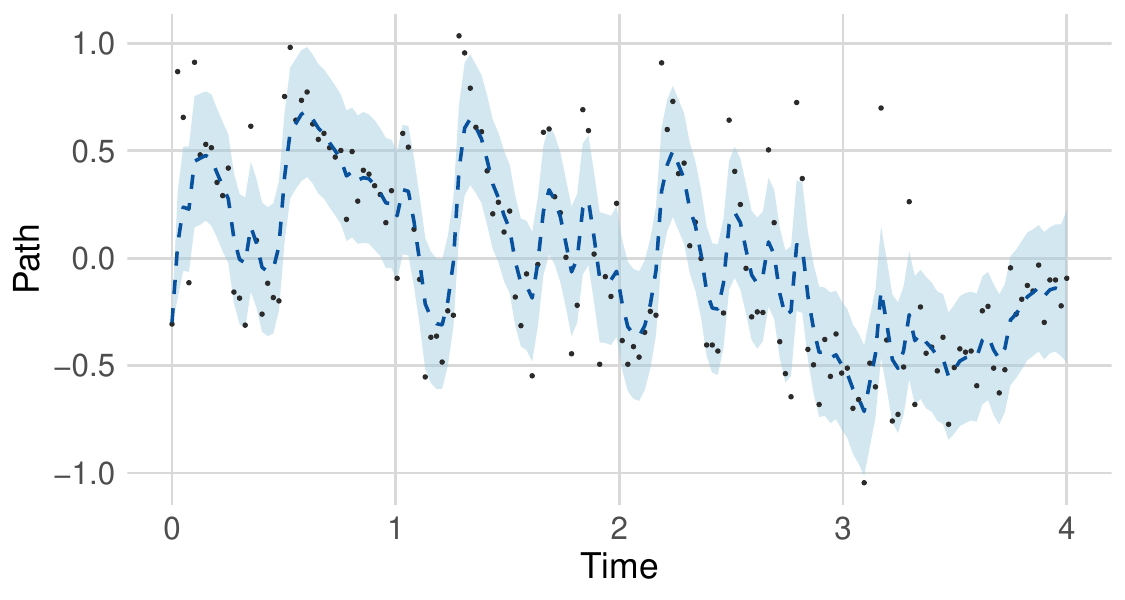}
    \end{minipage}
    \caption{
    (left) raw, and (center) preprocessed $\delta^{18}\text{O}$ data, along with (right) posterior median and $95\%$ credible interval over paths from our methodology}
    \label{fig:app_icecore_data}
\end{figure}
We follow~\citet{Kwasniok2009icecore} and use a bistable diffusion process to describe transitions between cold stadial and warm interstadial states. We choose %
the double-well potential Langevin process described by the SDE below:
\[
 \hspace*{1in}    dX_t = (-pX_t^3 + qX_t)dt + dW_t,\quad X_0 = -0.3,\quad t\in [0,4]  %
    .\]
We place independent exponential priors with rate 2 on $\theta = (p,q)$, and  
a weakly-informative conjugate prior on the noise variance, $\sigma_N^2 \sim \operatorname{IG}(10^{-3},10^{-3})$, where $\operatorname{IG}(a,b)$ denotes an inverse-gamma distribution with shape parameter $a$ and rate parameter $b$.
We performed posterior inference using the \gibbspost\ with joint path-parameter updates following the configuration of the previous experiments.
\begin{figure}[tbh!]
    \centering
    \begin{minipage}{0.3\textwidth}
        \centering
        \includegraphics[width=\linewidth]{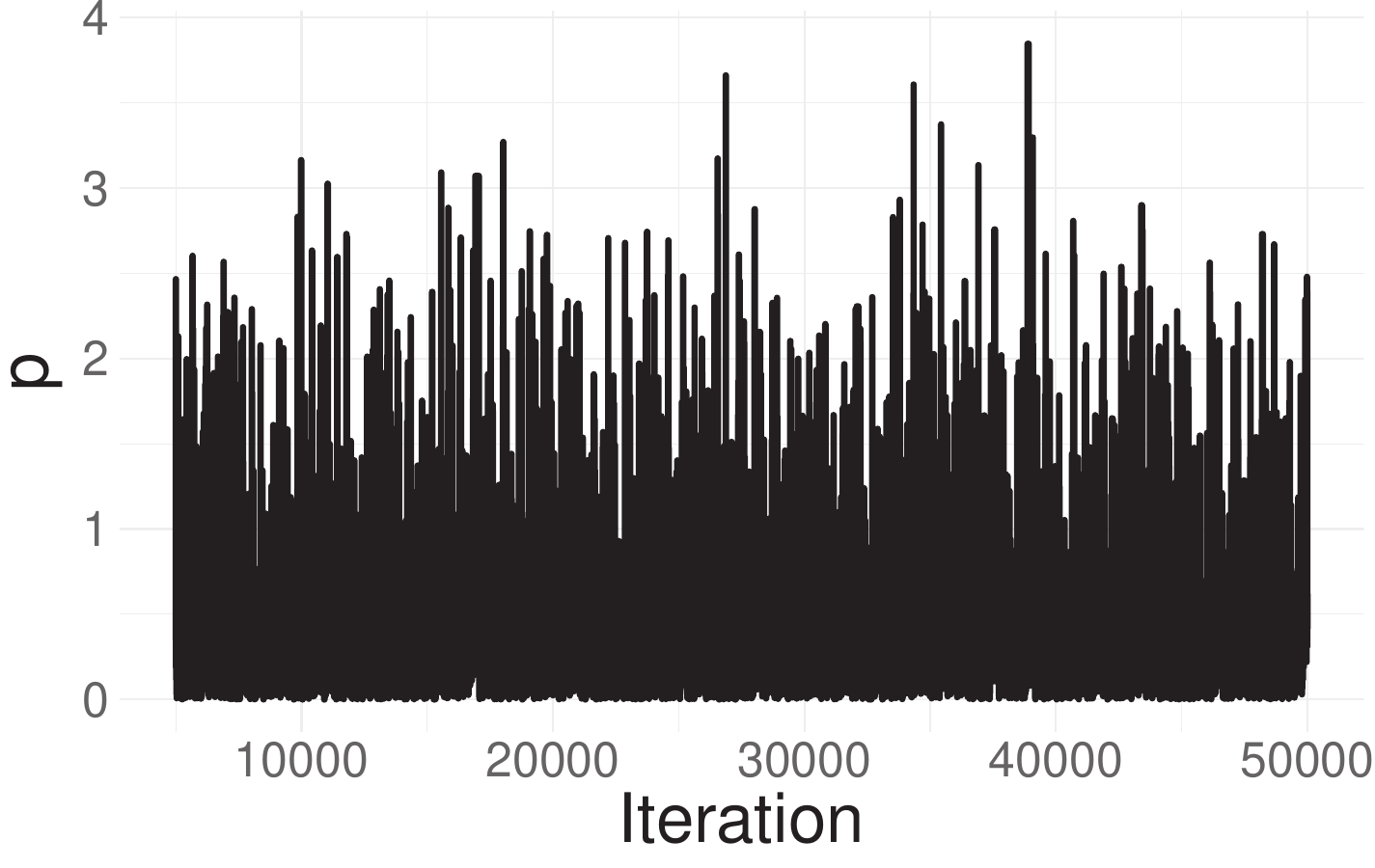}
    \end{minipage}
    \hfill
    \begin{minipage}{0.3\textwidth}
        \centering
        \includegraphics[width=\linewidth]{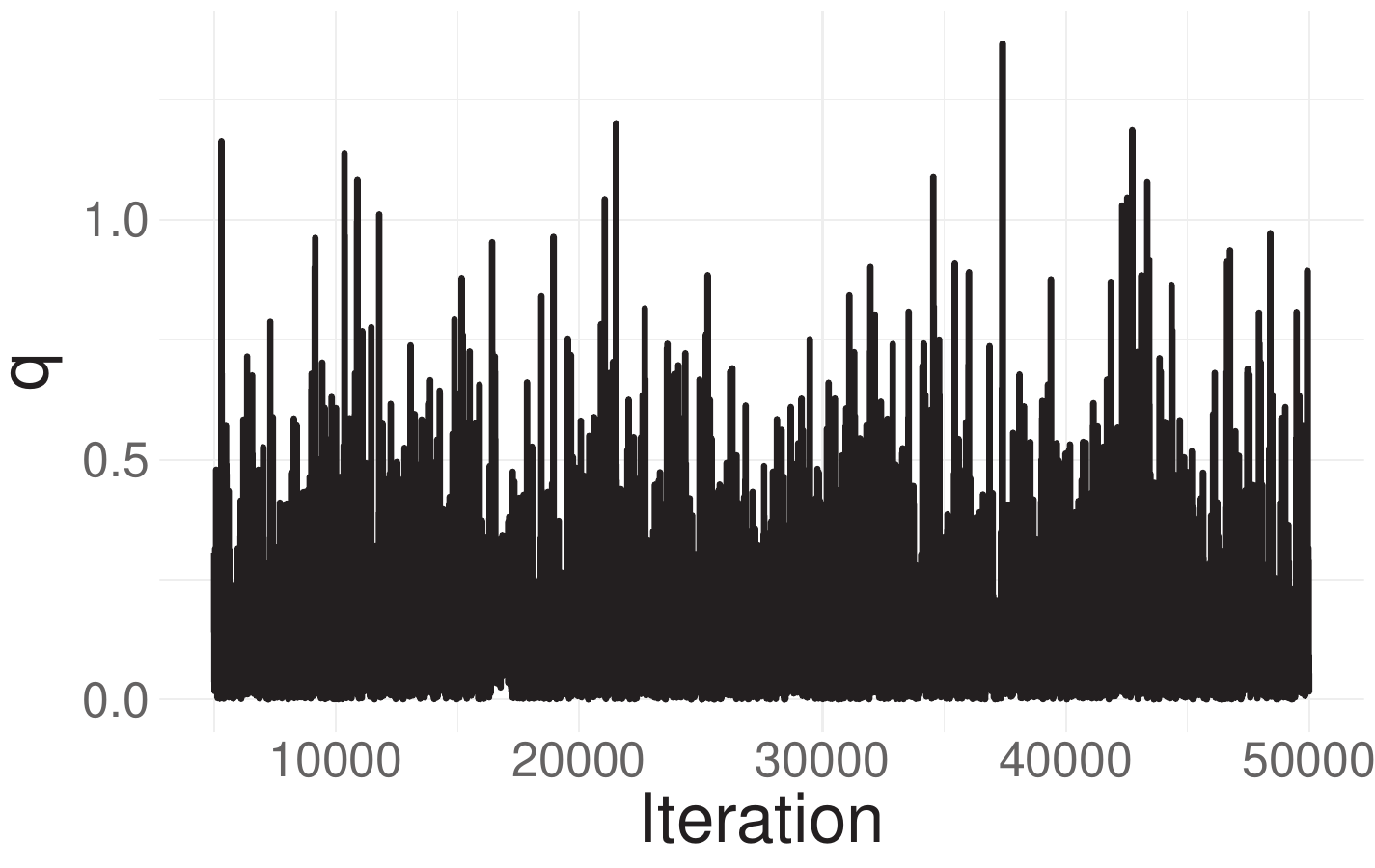}
    \end{minipage}
    \hfill
    \begin{minipage}{0.3\textwidth}
        \centering
        \includegraphics[width=\linewidth]{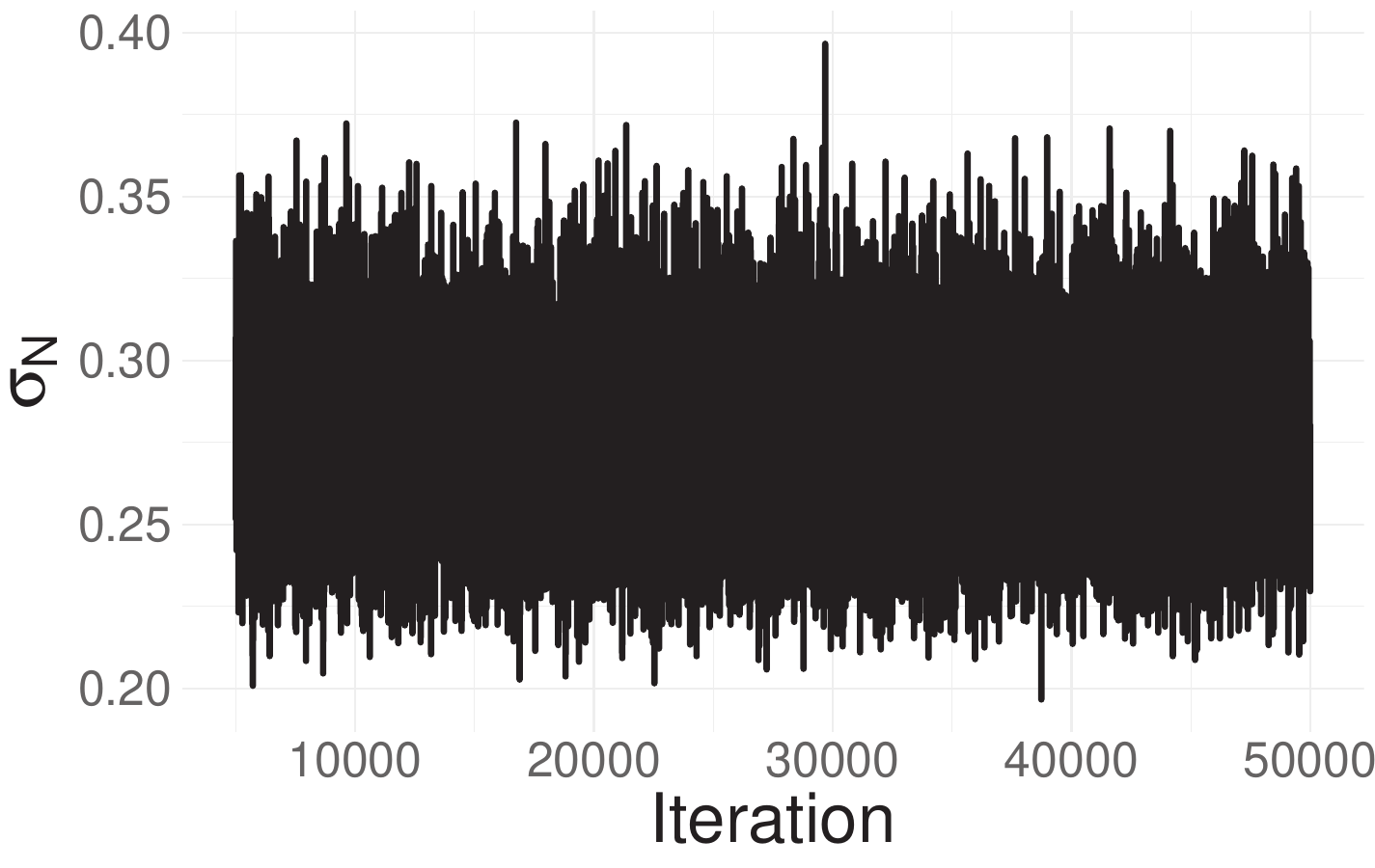}
    \end{minipage}
    \begin{minipage}{0.3\textwidth}
        \centering
        \includegraphics[width=\linewidth]{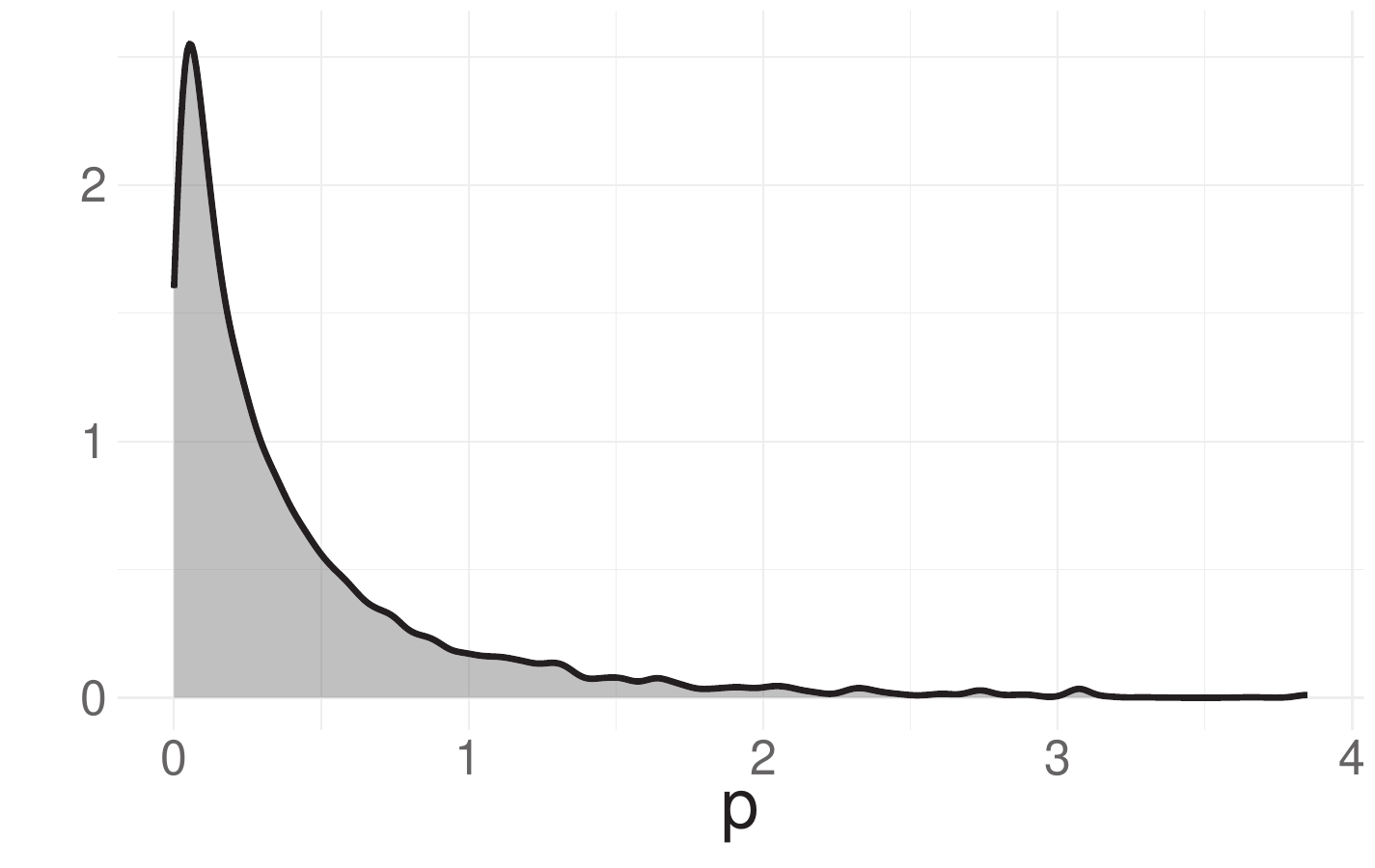}
    \end{minipage}
    \hfill
    \begin{minipage}{0.3\textwidth}
        \centering
        \includegraphics[width=\linewidth]{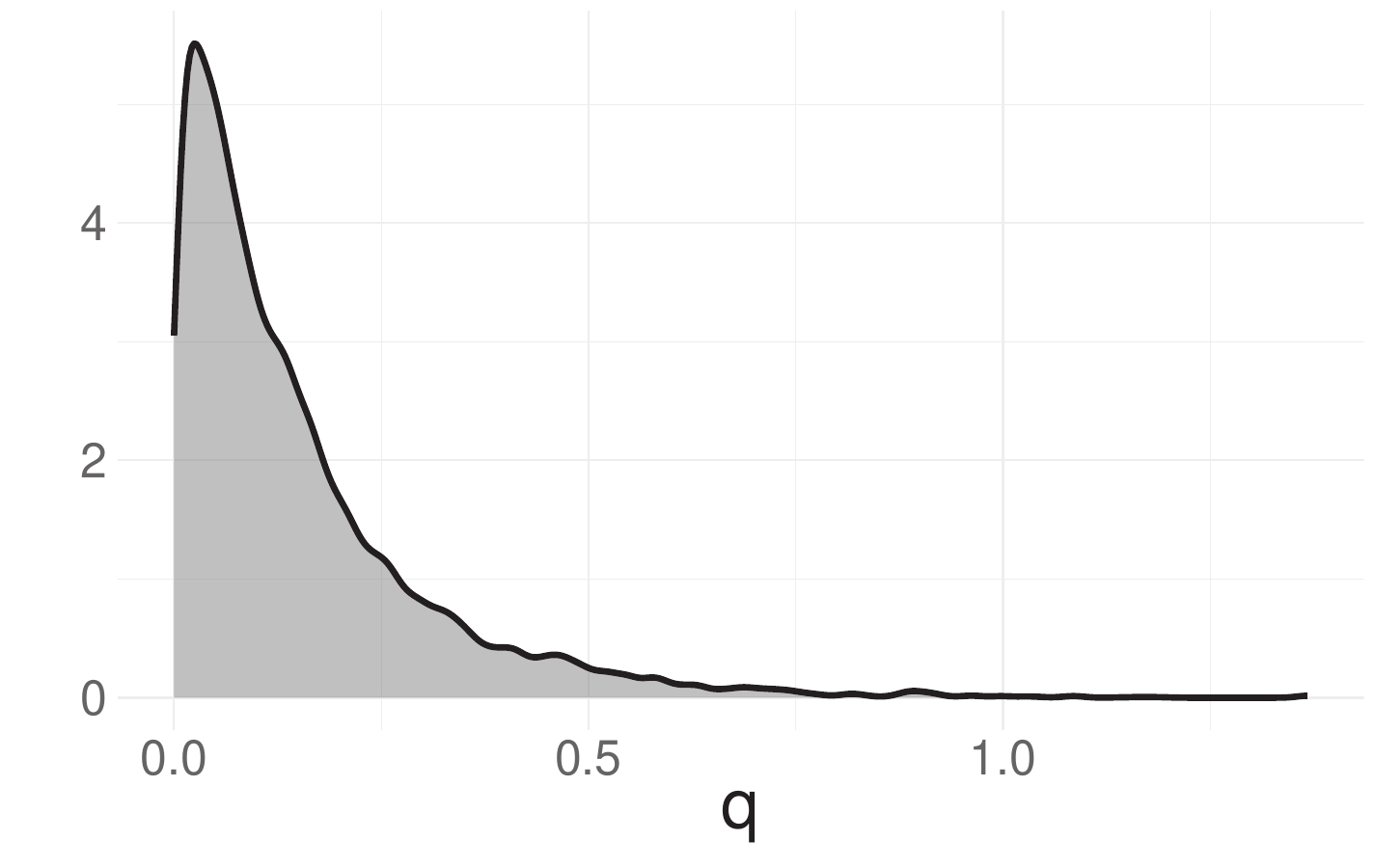}
    \end{minipage}
    \hfill
    \begin{minipage}{0.3\textwidth}
        \centering
        \includegraphics[width=\linewidth]{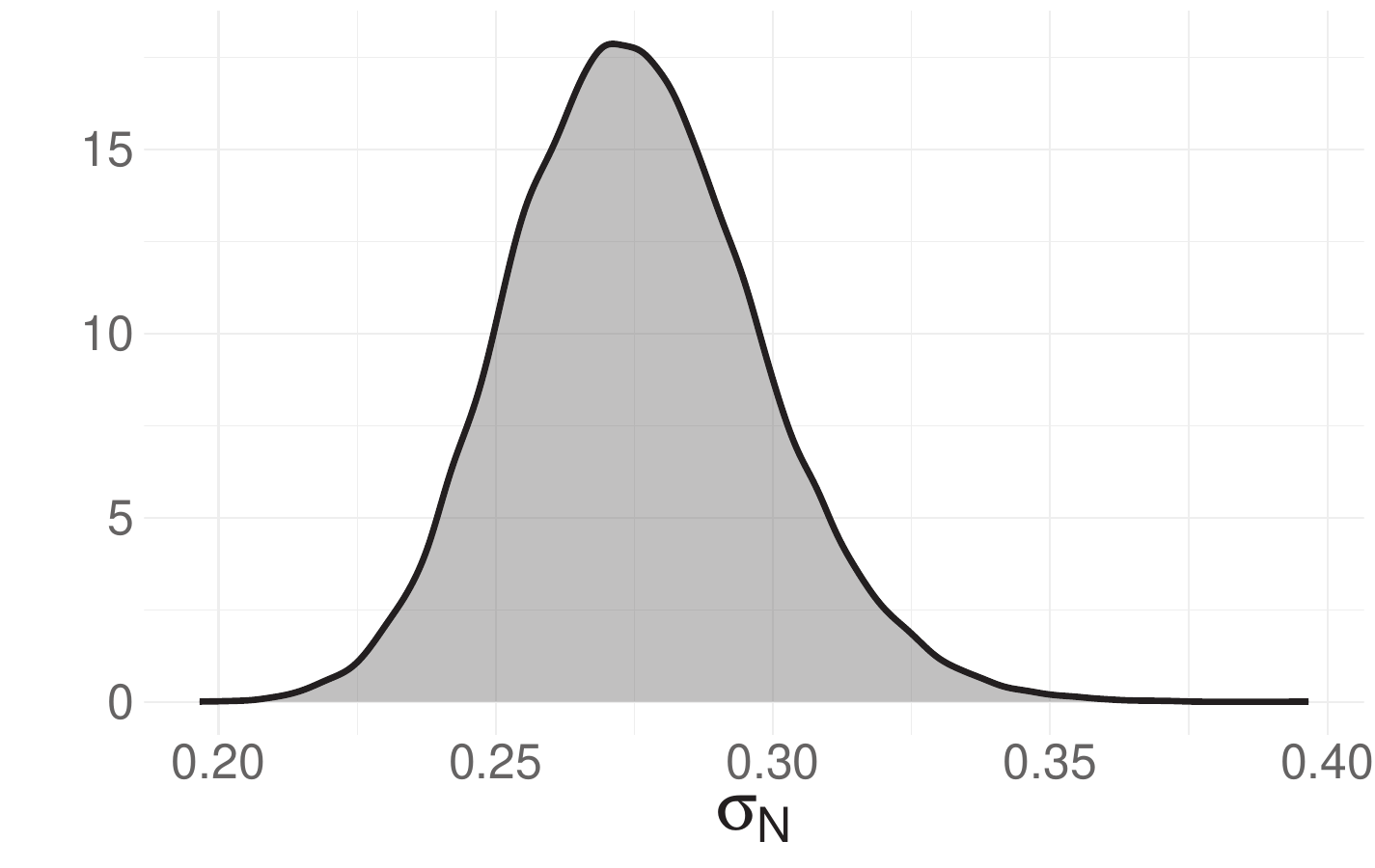}
    \end{minipage}

    \caption{Posterior inference results for the ice-core data. Top row: trace plots of $(p,q,\sigma_N)$. Bottom row: marginal posterior distributions.}
    \label{fig:app_icecore_trace_marginal}
\end{figure}

Our algorithm yielded acceptable mixing of the parameters as shown in Figure~\ref{fig:app_icecore_trace_marginal}.
The posterior distribution of the latent diffusion path is summarized by its median and a $95\%$ credible interval in Figure~\ref{fig:app_icecore_data} (right). The resulting credible band encompasses a substantial fraction of the observations, indicating that the double-well potential Langevin model with additive Gaussian noise provides a reasonable description of the underlying dynamics in the ice-core record.
The maximum a posteriori estimates (MAPs) of $p$, $q$, and $\sigma_N$ are given by $0.0574$, $0.0247$, and $0.2712$, respectively. A quantity of interest is the MAP of the separation between the two potential wells (given by $2\sqrt{q/p}$), and from our simulations, this has value $0.6998$. %

\section{Discussion}
We proposed a flexible framework for MCMC for diffusions of class EA2 and EA3. MCMC is generally unavoidable in the setting of posterior simulation, and by directly focusing on this, our methodology is significantly simpler and more flexible than those based on exact rejection sampling. Our framework is easily customizable to specific applications, e.g.\ the HMC step can be replaced by appropriate data-augmentation schemes for probit, logistic, or point process observations.
We also indicated how our methodology can be useful for parameter inference by reducing the coupling between path and parameter.
\texttt{R} code implementing our methodology is publicly available at the repository: \url{https://github.com/minhyeokstat/GibbsUEA}.
Some directions for future investigation include extending our methodology to jump diffusions, considering situations where the parameter also affects the diffusion term, better understanding the ergodicity properties of our sampler both in path and parameter space, and finally, developing modeling methodology that can exploit the availability of efficient sampling algorithms.

\bibliographystyle{abbrvnat}
\bibliography{ref}

\pagebreak\appendix
\section{Proofs}

\setcounter{proposition}{2}  %

{
\begin{proposition}
With a prior $\ptheta(\theta)$ on $\theta$, the joint law of $\theta, X, \psi$ and $Y$ is given by
\begin{align}
    \mathcal{P}(d\theta,dX,d\psi) 
        \propto & \ptheta(d\theta)l_Y(X)\cZ_{h_{\theta}}(dX)\mathcal{M}_{\Delta_\theta(X)}(d\psi)\exp(-m_\theta(X)\cdot T) \nonumber \\
        & \prod_{g\in \psi}\left(\frac{M_\theta(X)-\phi_\theta(X_{g})}{\Delta_\theta(X)}\right)\times
        {\exp\left\{-\alpha_\theta^{\downarrow}T+\log c_{\theta}\right\}}. \label{eq:app_joint_appx}
\end{align}
\end{proposition}

\begin{proof}
    Throughout the proof, we keep all $\theta$-dependent terms explicit. Our starting point is the Radon-Nikodym derivative between the diffusion measure $\cQ_{\theta}$ and $\cZ$ \citep{beskos2005exact}:
    \begin{align*}
        \frac{d\cQ_\theta}{d\cZ}(X) &= \exp\left\{A_\theta(X_{T})-A_\theta(X_{0}) - \frac{1}{2}\int_{0}^{T}\left(\alpha_\theta^2(X_{s}) + \alpha_\theta^{\prime}(X_{s})\right)ds\right\}\\
        &= \exp\left\{A_\theta(X_{T})-A_\theta(X_{0}) -\alpha_\theta^{\downarrow}T- \int_{0}^{T}\phi_\theta(X_{s})ds\right\}.
    \end{align*}
     Now, we replace the reference measure of $X$ from $\cZ$ to $\cZ_{h_{\theta}}$. From its definition, the Radon-Nikodym derivative between $\cZ_{h_{\theta}}$ and $\cZ$ can be written as
    \begin{align*}
        \frac{d\cZ_{h_{\theta}}}{d\cZ}(X) &= \frac{h_{\theta}(X_{T}, X_{0})}{h^0(X_0)N(X_{T}\mid X_{0},T)}\cdot\frac{\cZ_{h_{\theta}}(dX \mid X_{0},X_{T})}{\cZ(dX \mid X_{0},X_{T})}
        = \frac{h_{\theta}(X_{T}, X_{0})}{h^0(X_0)N(X_{T}\mid X_{0},T)}\\
        &= \frac{\sqrt{2\pi T}}{c_{\theta}}\exp\left(A_\theta(X_{T})-A_\theta(X_{0})\right)
    \end{align*}
    Putting these two equations together gives
    \begin{align*}
        \frac{d\cQ_\theta}{d\cZ_{h_\theta}}(X) &= \frac{c_{\theta}}{\sqrt{2\pi T}} \exp\left\{ -\alpha_\theta^{\downarrow}T- \int_{0}^{T}\phi_\theta(X_{s})ds\right\}.
    \end{align*}
        Now, as in~\cref{eq:gibbs_target}, define a measure on the joint path-event space as: %
    \begin{align}\hspace*{-.2in}
       \frac{d  \ \ \mathcal{P}_\theta(X,\psi)\hfill}{d \left(\cZ_{h_\theta}(X)\times \mathcal{M}_{\Delta_\theta(X)}(\psi)\right)} = & %
       \frac{c_{\theta}}{\sqrt{2\pi T}} \exp\left\{ -\alpha_\theta^{\downarrow}T-m_\theta(X)\cdot T\right\} \nonumber \\
       & \prod_{g \in \psi} \left(\frac{M_\theta(X) -  \phi_\theta(X_g)}{\Delta_\theta(X)}\right)
    \end{align}
From the Campbell's theorem~\citep{kingman1992poisson}, we have that
    \begin{align}
        \mathbb{E}_{\psi\sim \mathcal{M}_{\Delta_{\theta}(X)}}&\left[\prod_{g\in \psi}\left( \frac{M_\theta(X) - \phi_{\theta}(X_{g})}{\Delta_{\theta}(X)}\right)\right]    \nonumber\\ & = \mathbb{E}_{\psi\sim \mathcal{M}_{\Delta_{\theta}(X)}}\left[\exp \sum_{g\in \psi}\log\left(1 - \frac{\phi_{\theta}(X_{g})- m_\theta(X)}{\Delta_{\theta}(X)}\right)\right] \nonumber\\
        &= \exp\left(m_\theta(X)\cdot T -\int_{0}^{T}\phi_{\theta}(X_{s})ds\right)  %
        \label{eq:app_campb}
    \end{align}   
This shows that $\cP_\theta(dX,d\psi)$ has $\cQ_\theta(dX)$ as its marginal. Since $\cQ_\theta$ is a probability measure, $\cP_\theta$ itself is a probability measure, with all constants accounted for.
Multiplying both sides by $\ptheta(\theta)l_Y(X)$, we have that the joint over $(\theta,X,\psi,Y)$ equals
\begin{align}
    P(dX,d\psi,d\theta,Y) & = \ptheta(\theta)l_Y(X)\cZ_{h_\theta}(dX)\mathcal{M}_{\Delta_\theta(X)}(d\psi) \frac{c_{\theta}}{\sqrt{2\pi T}} \nonumber \\
    & \quad \exp\left\{-\alpha_\theta^{\downarrow}T-m_\theta(X)\cdot T\right\}  \prod_{g \in \psi} \left(\frac{ M_\theta(X) - \phi_\theta(X_g)}{\Delta_\theta(X)}\right)
    \end{align} \label{eq:app_joint}
Since we have defined $\cP(dX,d\psi,d\theta) = P(dX,d\psi,d\theta,Y)$, this proves our result.
\end{proof}

\setcounter{theorem}{0}  %

\begin{theorem}
Define a probability $\cP$ on the path-point process product space as
\begin{equation}
    \label{eq:gibbs_target_appx}
    \cP(dX,d\psi) \propto \cZpost_{h}(dX) \cdot \mathcal{M}_{\Delta(X)}(d\psi) \cdot \exp(-m(X)T) \prod_{g \in \psi} \left(\frac{M(X) - \phi(X_g)}{\Delta(X)}\right).
\end{equation}
 Then, the probability measure $\cQpost$ (i.e.\ the SDE posterior) is the marginal of $\cP$: %
$
\cQpost(dX) %
\propto \int_{\psi} \cZpost_{h}(dX) \times \exp(-m(X)T) \times \mathcal{M}_{\Delta(X)}(d\psi) \times \prod_{g \in \psi} \left( \frac{M(X) - \phi(X_g)}{\Delta(X)}\right).
$
\label{thrm:joint_appx}
\end{theorem}
\begin{proof}

Consider~\cref{eq:app_joint_appx} from the previous proof. We keep $\theta$ fixed, and drop terms that do not depend on $X$ or $\psi$. Then, along with $\cZ^+_h(dX) \propto \cZ_h(dX) l_Y(X)$, gives
\begin{align*}\hspace{-.2in}
\cP_\theta(dX,d\psi) \propto \cZ^+_{h_\theta}(dX)\mathcal{M}_{\Delta_\theta(X)}(d\psi)  \exp\left\{- m_\theta(X)T \right\} \prod_{g \in \psi} \left(\frac{ M_\theta(X) - \phi_\theta(X_g)}{\Delta_\theta(X)}\right).
\end{align*}
    From~\cref{eq:app_campb}, we have 
    \begin{equation*}
        \exp\left\{- m_\theta(X)T \right\}\mathbb{E}_{\psi\sim \mathcal{M}_{\Delta_{\theta}(X)}}\left[\prod_{g\in \psi}\left(\frac{M_{\theta}(X)-\phi_{\theta}(X_{g})}{\Delta_{\theta}(X)}\right)\right] = \exp\left(-\int_{0}^{T}\phi_{\theta}(X_{s})ds\right).
    \end{equation*}   
Comparing with~\cref{eq:RN_der_BB} shows  $\cP_\theta(dX,d\psi)$ has $\cQpost_\theta(dX)$ as its marginal for any $\theta$.
    
\end{proof}

}
\setcounter{proposition}{0}
\begin{proposition}
Write $\cP(\cdot \mid \psi)$ for the conditional law over paths under~\eqref{eq:gibbs_target_appx}, for any $\psi$. Then this has density with respect to \(\mathcal{Z}_h\) given by:
\begin{align}
\frac{d \cP(X \mid \psi)}{d\mathcal{Z}_h} \propto 
l_Y(X_S) \cdot \exp(- M(X) \cdot T) \cdot \prod_{g \in \psi} \left({M(X)} - {\phi(X_g)}\right). 
\end{align}
\end{proposition}
\begin{proof}
This conditional follows directly from the joint probability in~\Cref{thrm:joint_appx}, and the fact that $\frac{d \mathcal{M}_{\Delta(X)}}{d \mathcal{M}_1}(\psi) \propto \exp(-\Delta(X)T)\Delta(X)^{|\psi|}$.
\end{proof}

\begin{proposition}
Write $\cP(\cdot \mid X)$ for the conditional law of $\psi$ given $X$ under~\eqref{eq:gibbs_target_appx}. %
For any $X$, this is the law of a Poisson process on \([0, T]\) with rate \(\lambda_X(t) = M(X) - \phi(X_t)\).
\end{proposition}
\begin{proof}
Substituting $\frac{d \mathcal{M}_{\Delta(X)}}{d \mathcal{M}_1}(\psi) \propto \exp(-\Delta(X)T)\Delta(X)^{|\psi|}$~\citep{reiss2012course} into the joint from~\Cref{thrm:joint_appx}, we have
    \begin{align*}
\frac{d\cP(\psi \mid X)}{d\mathcal{M}_1} & \propto  \exp(-M(X)\cdot T) \prod_{g \in \psi} \left({M(X) - \phi(X_g)}\right) \\
& \propto  \exp\left(-\int_{0}^{T} (M(X)-\phi(X_t))dt\right) \prod_{g \in \psi} \left({M(X) - \phi(X_g)}\right) 
\end{align*}
This gives the desired result.
\end{proof}

\section{The \texttt{GibbsEA2} sampler for EA2 diffusions}
\label{sec:gibbs_ea2}
Recall from~\Cref{sec:ea} that for EA2 diffusions, the function $\phi(\cdot)$ only is bounded over intervals of the form $[x,\infty)$. As a consequence, the rejection sampling algorithm of~\citet{beskos2005exact} requires first instantiating the minimum $\Xmin$ of the Brownian bridge proposal. %
Not surprisingly, our MCMC algorithm involves updating $\psi, X_{\G}$ \emph{and} $\Xmin$. Here, we mention a few facts, see~\citet{pollock2016exact} for details.
\begin{enumerate}[(a)]
    \item\ Conditioned on its minimum $\Xmin = (\Xmin_v,\Xmin_t)$, a Brownian bridge decomposes into two independent Bessel bridges, one on either side of the minimum. Conditioned further on its values an independent set of times $G$, %
    the process now decomposes into independent Bessel bridges over segments $(g_i,g_{i+1})$ defined by successive elements of $\sort{G \cup \Xmin_t}$. %
    \item\ Simulating the above Bessel bridges at a finite set of times is straightforward, this involves simulating and combining 3 independent Brownian bridges. %
     On the other hand, calculating the log-probability density of a Bessel bridge, and its gradient at a finite set of times is much harder than for a Brownian bridge.
    \item\ Simulating the minimum of a Brownian bridge is straightforward matter of simulating from a uniform and an inverse-Gaussian density; see~\citep[Algorithm 12]{pollock2016exact}. To simulate the minimum given path values at $G$, we independently simulate the minima of Brownian bridges over successive intervals of $\sort{G}$ and take the minimum of these minima. This can easily be vectorized for efficiency. Calculating the log-probability and its gradient %
    is again challenging.
\end{enumerate}

\texttt{GibbsEA2} targets $\cP(dX, d\psi)$ by alternately updating $\psi$ and $(X_\G,\Xmin)$. %
\subsection{\texorpdfstring{Conditionally updating the Poisson grid from $\cP(d\psi\mid X)$}{Conditionally updating the Poisson grid}}
Following Proposition~\ref{prop:condpoi}, %
for EA2 diffusions, conditioned on the path value $X$, the grid $\psi$ follows an inhomogeneous Poisson process with rate $M(\Xmin) - \phi(X)$.
To simulate this inhomogeneous Poisson process, we use the thinning theorem, first simulating a homogeneous Poisson process $\vartheta$  with rate \( \Delta(\Xmin) \), and then keeping each point $e \in \vartheta$ with probability \( \frac{M(\Xmin) - \phi(X_e)}{\Delta(\Xmin)} \).

In practice, at the start of this step, we only have $X_{G}$ and $\Xmin$. From the latter, we calculate $\Delta(\Xmin)$ and simulate $\vartheta$. Following point (b) above, we then impute $X_\vartheta$ by simulating from the Bessel bridge $\cZ_h(X_\vartheta\mid \Xmin, X_G)$. %
Given this, we thin the elements of $\vartheta$ as above.
The points that survive the thinning procedure form the new Poisson grid $\psi'$, and the new MCMC state is $(\psi',X_{\Gtick},\Xmin)$. %

\subsection{\texorpdfstring{Conditionally updating the Brownian path from $\cP(dX\mid\psi)$}{Conditionally updating the Brownian path}} \label{sec:path_updt}
For the biased Brownian bridge $\cZ_h$, we can decompose the probability of any path as (probability of its endpoints) $\times$ (probability of its values on $\G$ given end-point values) $\times$ (probability of the time and value of its minimum given earlier values)  $\times$ (probability of the rest of the path given everything so far).
Thus,
$$\cZ_h(dX) = p_{\cZ_h}(X_{\G}) dX_{\G}\cdot\pmin(\Xmin\mid X_{\G}) d\Xmin \cdot \cZ_h(dX\mid \Xmin,X_{\G}).$$
This ordering differs from the EA2 rejection sampler, where, to simulate $\psi$, we must {\em first} simulate the minimum $\Xmin$. Then, $X_{\G}$ follows a Bessel process, while under our ordering, $X_{\G}$ follows a simpler Gaussian distribution. This is possible only because we are \emph{conditioning} on $\psi$.
Now we can simplify the conditional (\cref{eq:path_cond}) as
\begin{align}
    \cP(dX \mid \psi)  \propto p(X_\G, & \Xmin \mid  \psi, Y)  dX_\G d\Xmin \cZ_h(dX \mid \Xmin, X_\G), \text{ with } \label{eq:tgt_decom_ea2}\\
    p( X_{\G}, \Xmin \mid \psi, Y)  & \propto \left[p_{\cZ_h}(X_{\G}) \cdot  l_Y(X_S) \right]\cdot \pmin(\Xmin\mid X_\G) \cdot \nonumber \\ 
    & \quad    [\exp(- M(\Xmin)\cdot T) \cdot \prod_{g \in \psi} \left(M(\Xmin) - {\phi(X_g)}\right)] \nonumber \\
    &  \hspace{-.2in} := [\piY(X_{\G})] \cdot \pmin(\Xmin\mid X_{\G}) \cdot 
    [\piM({\Xmin}, X_\psi)].\nonumber
\end{align}
Above, the only nonstandard term is $\pmin(\Xmin\mid X_\G)$, though we saw this is easy to simulate from. The term $\piY(\cdot)$ corresponds to a standard GP posterior.

To simulate \(\cP(dX\mid\psi)\), we employ a Metropolis-Hastings (MH) scheme, proposing a new path \(X^*\) from some probability \(\mathcal{V}(dX^*\mid X, \psi)\), and accepting with the usual MH acceptance probability.
Along the lines of~\cref{eq:tgt_decom_ea2},  we use a proposal distribution that decomposes as follows:
\begin{align*} 
{\mathcal{V}(dX^*\mid X,\psi)} = q(X^*_{\G}\mid X_{\G} ) \pmin(\XminStar \mid X^*_{\G}) dX^*_{\G} d{\XminStar} \cZ(dX^*\mid {\XminStar},X^*_{\G}). %
\end{align*}
That is, we propose new values $X^*_{\G}$ conditioned on the old values $X_{\G}$. Conditioned on these (and independent of the old values), we propose a new minimum ${\XminStar}$, with the remaining path then following the appropriate Bessel bridge. 
The acceptance probability simplifies to
\begin{align}
 \hspace{-.1in}   \text{acc} = 
     \frac{d\cP(X^*\mid \psi)}{d\cP(X\mid \psi)} \cdot \frac{d\mathcal{V}(X\mid X^*)}{d\mathcal{V}(X^*\mid X)} = 
    \frac{\piY(X^*_{\G})  
    \piM(\XminStar,X^*_\psi)}{\piY(X_{\G}) \piM(\Xmin,X_\psi)} \cdot \frac{q(X_{\G}\mid X^*_{\G} )}{q( X^*_{\G}\mid  X_{\G} )}.
\end{align}
Notice that the $\pmin(\cdot)$ term cancels out above, and the only terms remaining from the target distribution are straightforward to evaluate. We are only left to find an efficient proposal $q( X^*_{\G}\mid  X_{\G} )$ to explore the $X_{\G}$ space.

{ In general, $\piY(X_S)$ is intractable, and we demonstrate how MCMC methodology like Hamiltonian Monte Carlo (HMC) can be brought to this setting. Our overall approach is only slightly more complicated than %
\citet{wang2020exact}.}

\noindent \textbf{A double Metropolis-Hastings scheme~\citep{liang2010double}}:
First note that the density $\piY(X_{\G})$ can be written as $p_{\cZ_h}(X_S) l_Y(X_S) p_{\cZ_h}(X_\psi \mid X_S)$, where the last term is just a Brownian bridge conditional. Accordingly, we choose $q( X^*_{\psi \cup S}\mid  X_{\psi \cup S} ) = q( X^*_S \mid  X_{S} )  p_{\cZ_h}(X_\psi \mid X_S)$. %
For $q( X^*_S \mid  X_{S} )$, we can choose any Markov kernel that targets $p_{\cZ_h}(X_S) l_Y(X_S)$. We use HMC: our target is the posterior distribution for a Gaussian prior with a non-conjugate (but differentiable) likelihood, a setting where HMC is known to excel. 
One iteration of HMC involves $L$ leapfrog steps of size $\epsilon$ followed by an accept/reject step, with $L, \epsilon$ and a mass-matrix $M$ parameters of the algorithm~\citep{neal2011mcmc}. While we cannot analytically evaluate the HMC proposal density $q(X^*_S \mid X_S)$, from the detailed balance condition of HMC, we have 
$
p_{\cZ_h}(X^*_S) l_Y(X^*_S) \cdot q(X_S \mid X^*_S)
=
p_{\cZ_h}(X_S) l_Y(X_S) \cdot
q(X^*_S\mid X_S)$.
With these facts, the acceptance probability simplifies to the easy to calculate form
\begin{align}
    \text{acc} 
    &= \frac{\piM(\XminStar,X_\psi^*)}{\piM(\Xmin,X_\psi)} = 
    \frac{ \exp(- M({\XminStar}) \cdot T) \cdot 
     \prod_{e \in \psi} \left({M({\XminStar})} - {\phi(X^*_e)}\right)}{ \exp(- M(\Xmin) \cdot T) \cdot 
     \prod_{e \in \psi} \left({M(\Xmin)} - {\phi(X_e)}\right)}. \label{eq:mh_acc}
\end{align}
Thus our scheme involves two accept/reject steps each iteration: one for the HMC proposal (to correct for only approximately simulating the Hamiltonian dynamics), and one to correct for the mismatch between the proposal and target distribution arising from the Poisson events. These details are repeated in Algorithm \ref{alg:ea2gibbsparainf}. %

\begin{algorithm}[ht]
   \caption{One iteration of \texttt{GibbsEA2} for EA2 diffusions.}
   \label{alg:ea2gibbsparainf}
   \begin{flushleft}
   \textbf{Input:}  
   \begin{itemize}[leftmargin=1.5em, labelsep=0.5em, itemsep=0pt, topsep=0pt, parsep=0pt, partopsep=0pt]
      \item Observation times $S \subset [0,T]$, along with likelihood function {$l_Y(X_S)$}. %
      \item Poisson times $\psi$, path values $X_{\G}$, and path minimum $\Xmin = (\Xmin_v, \Xmin_t)$.
   \end{itemize}
   \textbf{Output:}  
   \begin{itemize}[leftmargin=1.5em, labelsep=0.5em, itemsep=0pt, topsep=0pt, parsep=0pt, partopsep=0pt]
      \item Updated Poisson times $\psi{'}$, path values $X{'}_{\psi{'} \cup S}$ and minimum ${\Xminpr}$. %
   \end{itemize}
   \vspace{0.3em}
   \hrule
   \end{flushleft}
   \vspace{-0.5em}
   \begin{algorithmic}[1]
    \State \textbf{Update $\psi$ conditioned on $X$:}
      \State \quad Simulate a rate-\( \Delta(\Xmin) \) homogeneous Poisson process $\vartheta$ on $[0, T]$.
      \State \quad Impute $X_\vartheta$ according to Bessel bridge. %
      \State \quad Discard each $e \in \vartheta$ with probability $\frac{\phi(X_e)-m(\Xmin)}{\Delta(\Xmin)}$. Let $\psi'$ be the surviving times.
      \State \quad Retain only $\psi'$, $X_{\psi' \cup S}$, and ${\Xmin}$, and discard all other variables.
    \State \textbf{Update $X$ conditioned on $\psi$:}
      \State %
\quad Simulate $X_S^* \mid X_S$ from a Markov kernel  targeting $p_{\cZ_h}(X_{S})l_Y(X_S)$.
      \State \quad Simulate $X^*_{\psi'} \mid X^*_S$ from a Brownian bridge.
      \State \quad Simulate the minimum $\XminStar$ of a Brownian bridge given {$X^*_{\psi'\cup S}$}.
      \State \quad Set  {$(X'_{\psi'\cup S}, {\Xminpr})$} to {$(X^*_{\psi'\cup S}, {\XminStar})$} with probability given in~\cref{eq:mh_acc}.
\State \textbf{Return} the updated values {$(\psi', X'_{\psi'\cup S}, {\Xminpr}$)}.
   \end{algorithmic}
\end{algorithm}

\noindent \textbf{Why not directly use HMC to target the posterior conditional distribution $\cP(X \mid \psi)$, instead of the proposal distribution?} This was the approach taken in~\citet{wang2020exact} for EA1 diffusions. The EA2 setting requires instantiating the minimum of the Brownian bridge, requiring us to target a Bessel bridge posterior, a significantly more complex object %
whose log-probability and gradient are much more challenging to evaluate than a Brownian bridge posterior. %

In work not reported here, we used numerical differentiation to carry out this approach. In addition to being slow, targeting $(X_{\G} \mid \Xmin)$, requires that $\Xmin_v < \min X_{\G}$ is always satisfied. While the exact Hamiltonian dynamics will respect this, the leapfrog approximation resulted in this being violated very frequently.

\section{Experiments with MCMC parameters}

\subsection{Experiments on HMC Tuning}

One of the goals of this paper is to bring tools such as Hamilton Monte Carlo to the setting of SDEs. HMC itself comes with a number of hyperparameters, and it is important to explore how robust it is to their settings. The three key hyperparameters influencing the mixing behavior and efficiency are the leapfrog step size \(\epsilon\), the number of leapfrog steps \(L\), and the $(|S|) \times (|S|)$ mass matrix \(M\). We recall again that our HMC sampler operates on SDE values at a fixed set times, the observation times $S$ (which includes $0$ and $T$). This is in contrast to the algorithm of~\citet{wang2020exact} that also operates on the Poisson grid $\Psi$. Since $\Psi$ changes from iteration to iteration, tuning $M$ is much simpler in our case.

We tested a range of values for the step size \(\epsilon\) (\(\{0.05, 0.1, 0.2, 0.4\}\)) and the number of leapfrog steps \(L\) (\(\{2, 5, 10\}\)). We set the mass matrix \(M \in \{0.2M_H, M_H, 5M_H\}\), where $M_H$ is the Hessian of the log-posterior distribution at its mode.
We evaluated these settings using EA2 example described above,  %
with 10 observations under Gaussian noise with \(\sigma_N = 0.1\), and over 3 settings of the time interval: \(T \in \{1, 5, 10\}\). 
For each setting, we performed 10 repetitions of \gibbshmc\  and averaged the resulting metrics. Again, we used ESS/s as the performance metric, and summarize the results in~\Cref{fig:tunehmc_gnex3sn}.

\begin{figure}
 \centering
 \caption{ESS/s for different HMC settings in EA2 example with Gaussian noise $\sigma_N = 0.1$}
 \includegraphics[scale = 0.4]{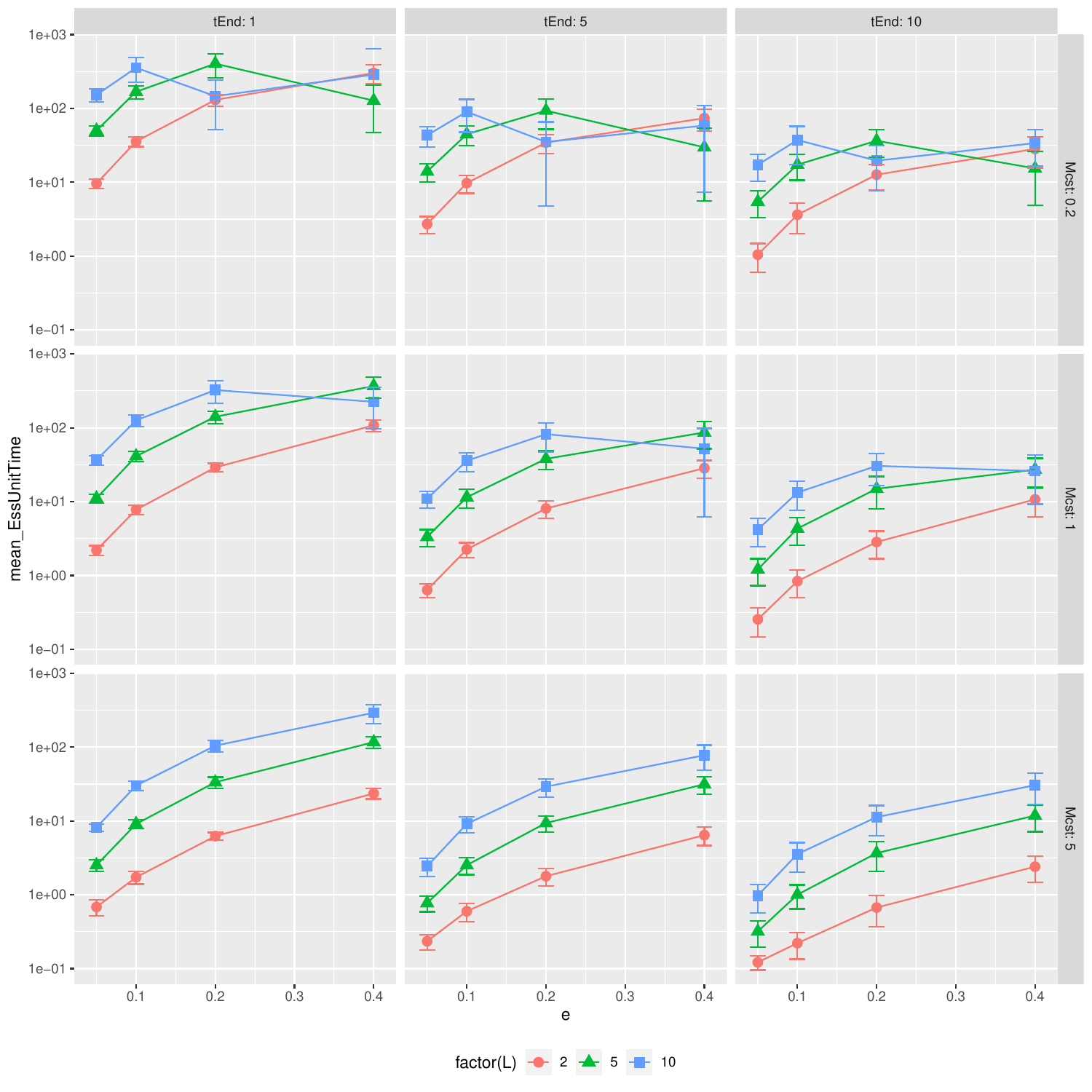}
 \label{fig:tunehmc_gnex3sn}
\end{figure}

There are several conclusions we can take from the figure.
Generally, a larger number of leapfrog steps yields better ESS while also requiring longer runtime. Accounting for both using ESS/s, we see that \(L = 10\) outperforms \(L = 5\), and \(L = 5\) is superior to \(L = 2\). Performance plateaus with larger values of $L$, though we do not include this. 
While the optimal settings of the hyperparameters fluctuate across different settings of $T$, for large enough $L$ (e.g.\ 5 or 10), performance does not fluctuate too dramatically.
Similarly, the different settings of $M$ that we considered also do not affect performance too dramatically.
Overall, we recommend \((\epsilon, L, M) = (0.1, 10, M_H)\): while this performs slightly worse than the best settings across all time lengths, it corresponds to standard default settings, and always gives competitive performance. We use this setting all our other experiments.

\subsection{Experiments on the Auxiliary Poisson rate \texorpdfstring{$\AuxRate$}{AuxRate}}

The \texttt{GibbsEA3} algorithm requires specification of the auxiliary Poisson events rate, $\AuxRate$. Choosing $\AuxRate$ too small would result in slow mixing of $\psi$, whereas excessively large values increase computational cost without commensurate gains in efficiency. We suggest a guideline for choosing $\AuxRate$, while acknowledging that comprehensive work remains to be done from both theoretical and empirical perspectives. 

Our suggestion is to choose $\AuxRate$ so that the average of allocation probability $\frac{M(\layer)-\phi(X_t)}{\AuxRate+ M(\layer)-\phi(X_t)}$ lies within a moderate intermediate range, for example between $0.4$ and $0.6$. Values in this range correspond to a balanced relabeling regime in which neither $\psi$ nor $\xi$ overwhelmingly dominates the configuration of $\psi \cup \xi$.

In practice, a suitable value of $\AuxRate$ may be obtained from a short pilot run of the algorithm or from samples collected during the burn-in period. One convenient choice is given by
\begin{equation*}
    \AuxRate \;\approx\; \frac{1}{N_{\operatorname{iter}}}\sum_{i=1}^{N_{\operatorname{iter}}} \left[M({\layer}^{(i)}) -  \frac{1}{|\psi^{(i)}|}\sum_{t\in \psi^{(i)}}\phi(X_{t})\right],
\end{equation*}
where ${\layer}^{(i)}$ and $\psi^{(i)}$ denote the layer and the set of auxiliary event times assigned to $\psi$ at the $i$th iteration, respectively.

To illustrate this choice, we consider a CIR process under the same setting as before, but with three different levels of time-window lengths, $T \in \{5,10,15\}$. For each value of $T$, we run the \texttt{GibbsPost} algorithm using ten different values of the auxiliary rate $\AuxRate$, ranging from $0.01$ to $10$. Figure~\ref{fig:tune_auxrate} shows the mean allocation probability $\frac{M(\layer)-\phi(X_t)}{\AuxRate+ M(\layer)-\phi(X_t)}$ computed for each time-window length by taking median over ten independent replications. Based on our guideline and the numerical results in Figure~\ref{fig:tune_auxrate}, a moderate range for $\AuxRate$ is between $1$ and $5$. We observe that satisfactory mixing of both the Poisson counts $|\psi|$ and the latent path is achieved for values of $\AuxRate$ within this range. Throughout the paper, we use the value of $\AuxRate = 2$.

\begin{figure}[tbh!]
    \centering
    \includegraphics[scale=0.45]{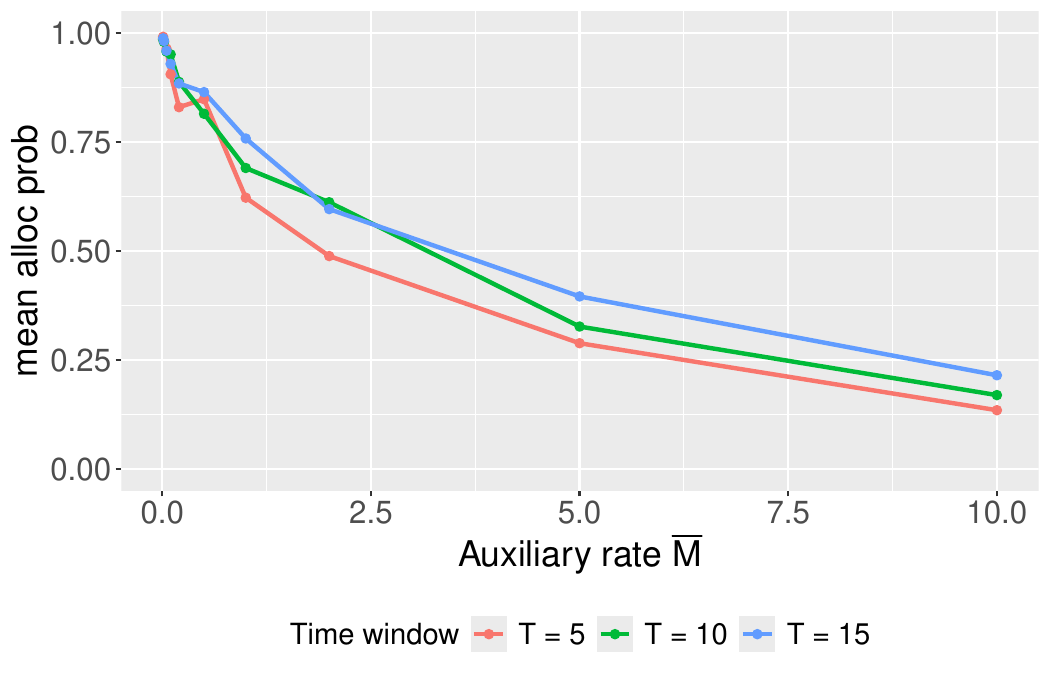}
    \caption{Mean allocation probability in relabeling step in \texttt{GibbsEA3} algorithm for the CIR process.}
    \label{fig:tune_auxrate}
\end{figure}

\section{Additional Experiments}
\label{syndata}

\subsection{EA2 example 2}\label{SDE2}
We consider the diffusion parameterized by $\theta = (p,q)$ with $p > 0$, $q > 0$:
\begin{align}
dX_t = - p X_t \exp(-q X_t) \, dt + dW_t. \label{eq:sde2}
\end{align}

\noindent \textbf{Gaussian and Poisson Observations}
We set the parameters $(p,q) = (1,0.1)$, and evaluate the different samplers using synthetic data following the same setup as EA2 example.
We plot the results in Figure \ref{fig:gnex4_ess_over_time_ln_theta0.1} and  \ref{fig:prex4_ess_over_time_theta0.1} for the case of Gaussian and Poisson observations respectively.
Now, compared to EA2 example, we observe smaller ESS/s values across all samplers, reflecting the fact that this is a harder problem. At the same time, the conclusions from these plots are largely the same as before, and reaffirm once again that \gibbshmc, through its flexibility and performance, is once again the best option.

\begin{figure}
 \centering
 \includegraphics[width=0.8\textwidth]{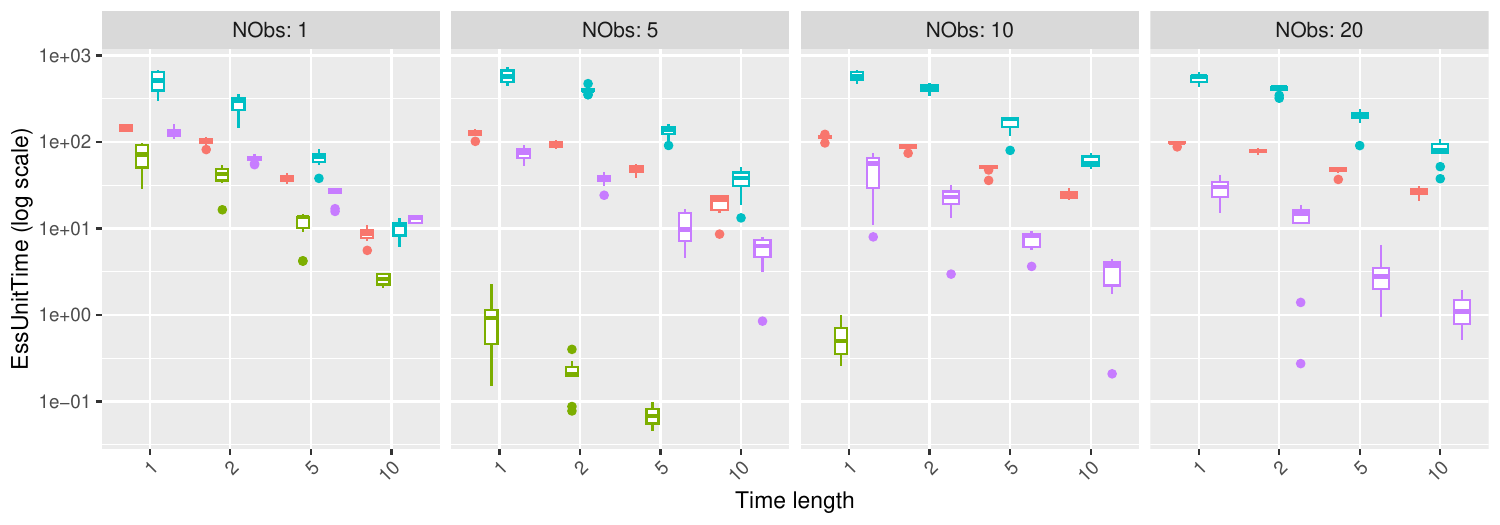}
 \centering
 \includegraphics[width=0.8\textwidth]{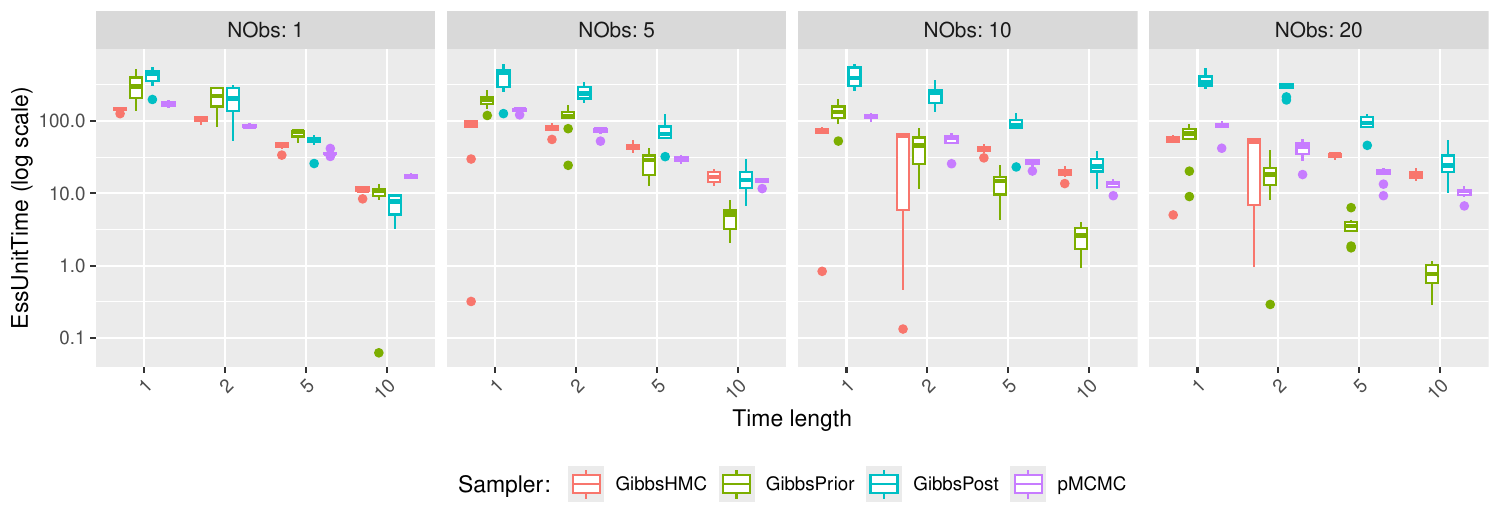}
\caption{The effective samples per seconds vs time lengths of EA2 example 2 under Gaussian noise $\sigma_N = 0.1$ (top row) and $\sigma_N = 1$ (bottom row)} \label{fig:gnex4_ess_over_time_ln_theta0.1}
\end{figure}

\begin{figure}
 \centering
 \caption{ESS/s vs time lengths of EA2 example 2 under Poisson response}
 \includegraphics[width=0.8\textwidth]{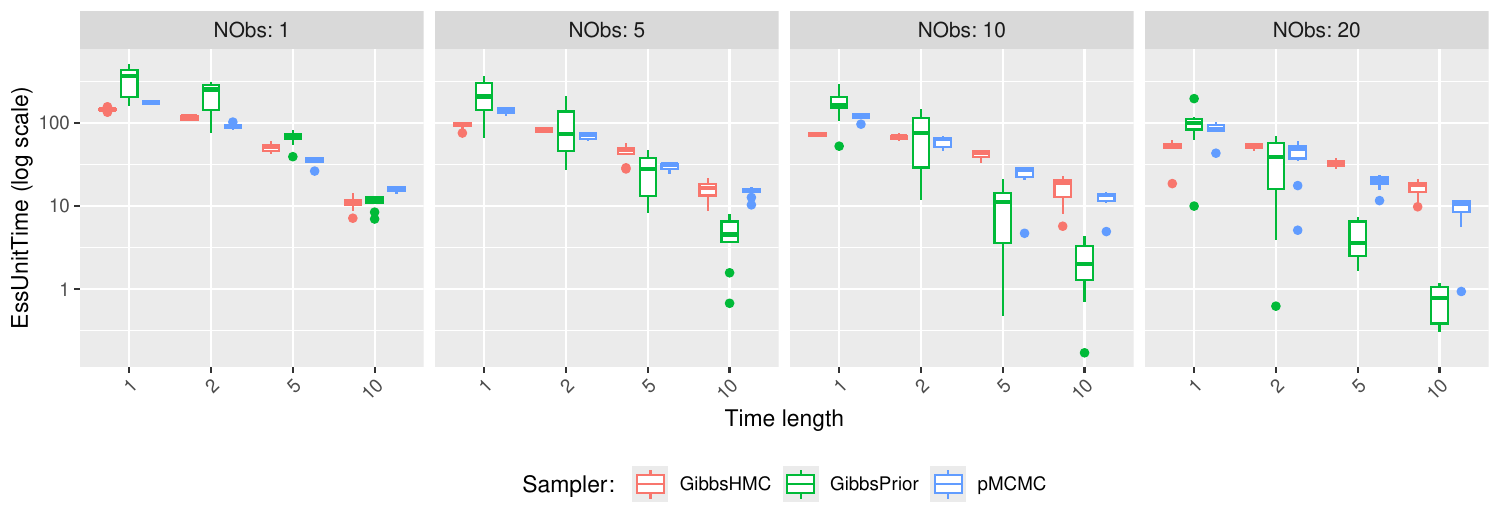}
 \label{fig:prex4_ess_over_time_theta0.1}
\end{figure}

\subsection{New York Temperature Data}\label{sec: ny_applications}
We consider an application involving a dataset of temperature measurements in New York city (NYC).
We downloaded the annual average temperature data for New York Central Park from 1870 to 2023 from the National Center for Environmental Information\footnote{available at {\tiny \url{https://www.ncei.noaa.gov/access/search/data-search/global-summary-of-the-year?pageNum=3&dataTypes=TAVG&bbox=38.075,-122.715,37.485,-122.125}}}. The dataset contains 154 data points in total.

To prepare the data, we first subtracting out the linear trend, and then rescaled the time length to $15$. We selected every other data point as observations, using the rest as a test dataset. 
Figure \ref{fig:app_NYtemperature} illustrates the original annual average temperature data, the detrended and rescaled data, and the subset used as observations. This plot shows substantial fluctuations in temperature from year to year.

With this as our data, we modeled the underlying temperature process as a realization of EA2 example 2 (\cref{eq:sde2}). We ran our algorithm both with the parameter $\theta$ fixed to $0.1$, and with a   uniform prior. 
\begin{figure}[H]
 \caption{Annual average temperature for New York Central Park from 1870 to 2023}
 \centering
 \includegraphics[scale=0.22]{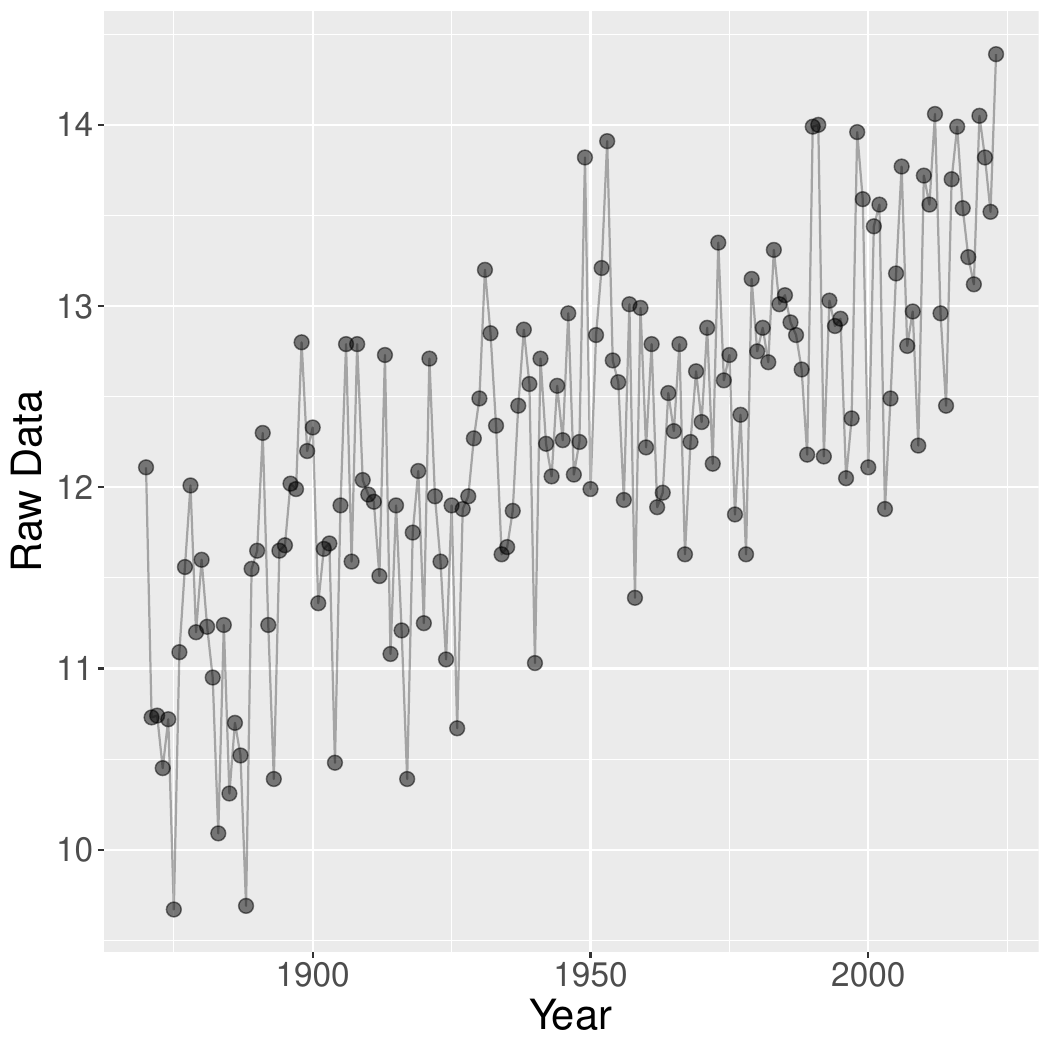}
 \includegraphics[scale=0.22]{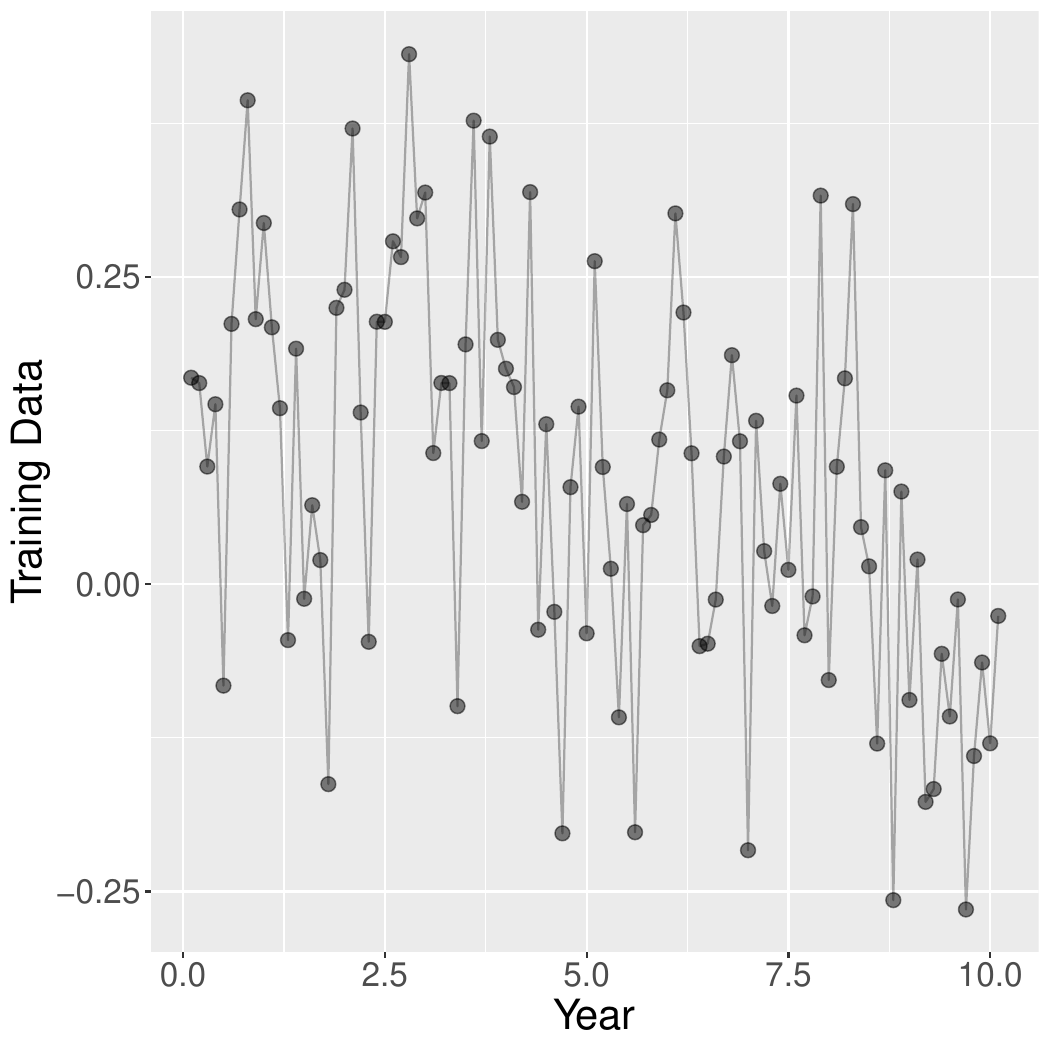}
 \label{fig:app_NYtemperature}
\end{figure}
Figure \ref{fig:app_NYtemperature_pstpath} shows  both the observed data as well as the posterior quantile bands. The posterior paths closely align with the unobserved data, demonstrating the good performance of our sampler.

\begin{figure}[H]
 \caption{90\% posterior intervals for fixed parameter and parameter inference cases}
 \centering
 \includegraphics[scale=0.25]{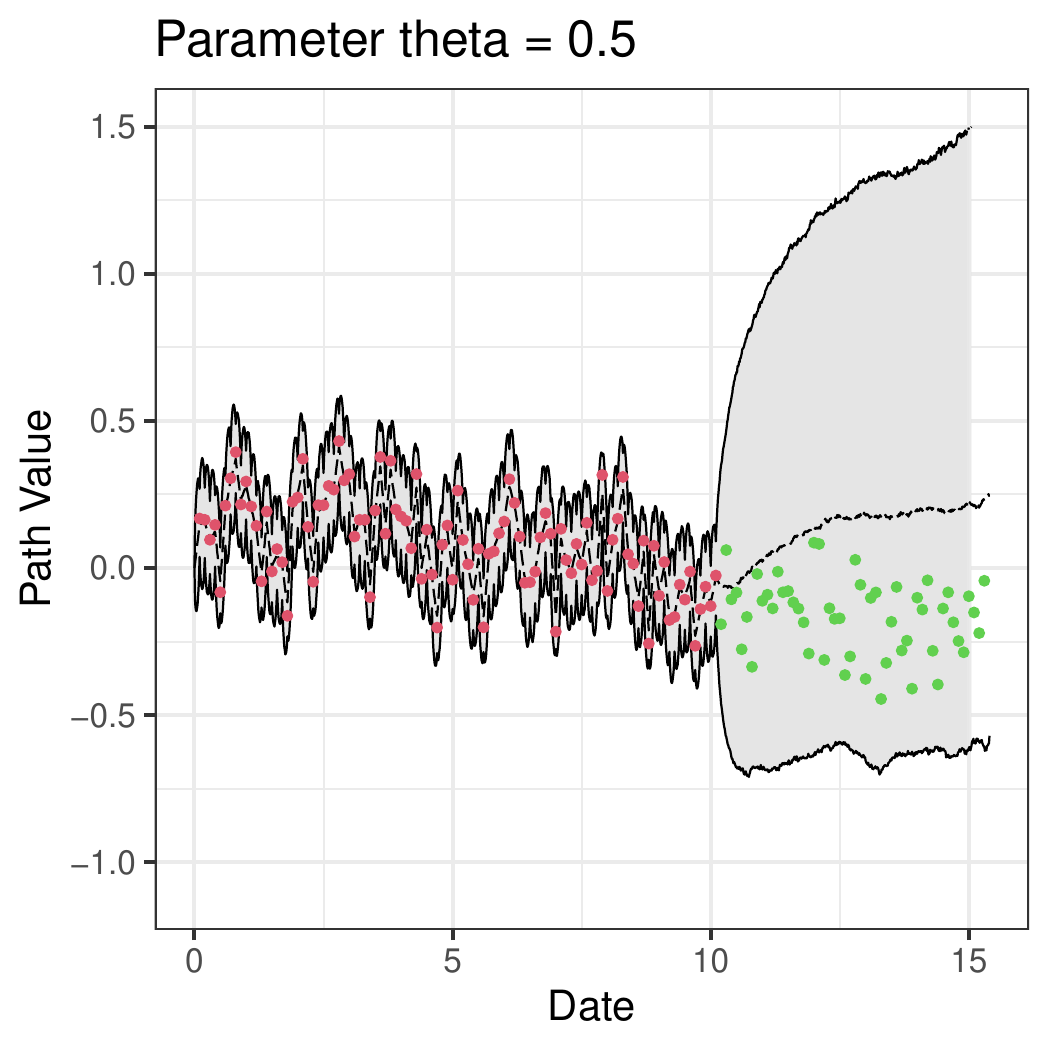}
  \includegraphics[scale=0.25]{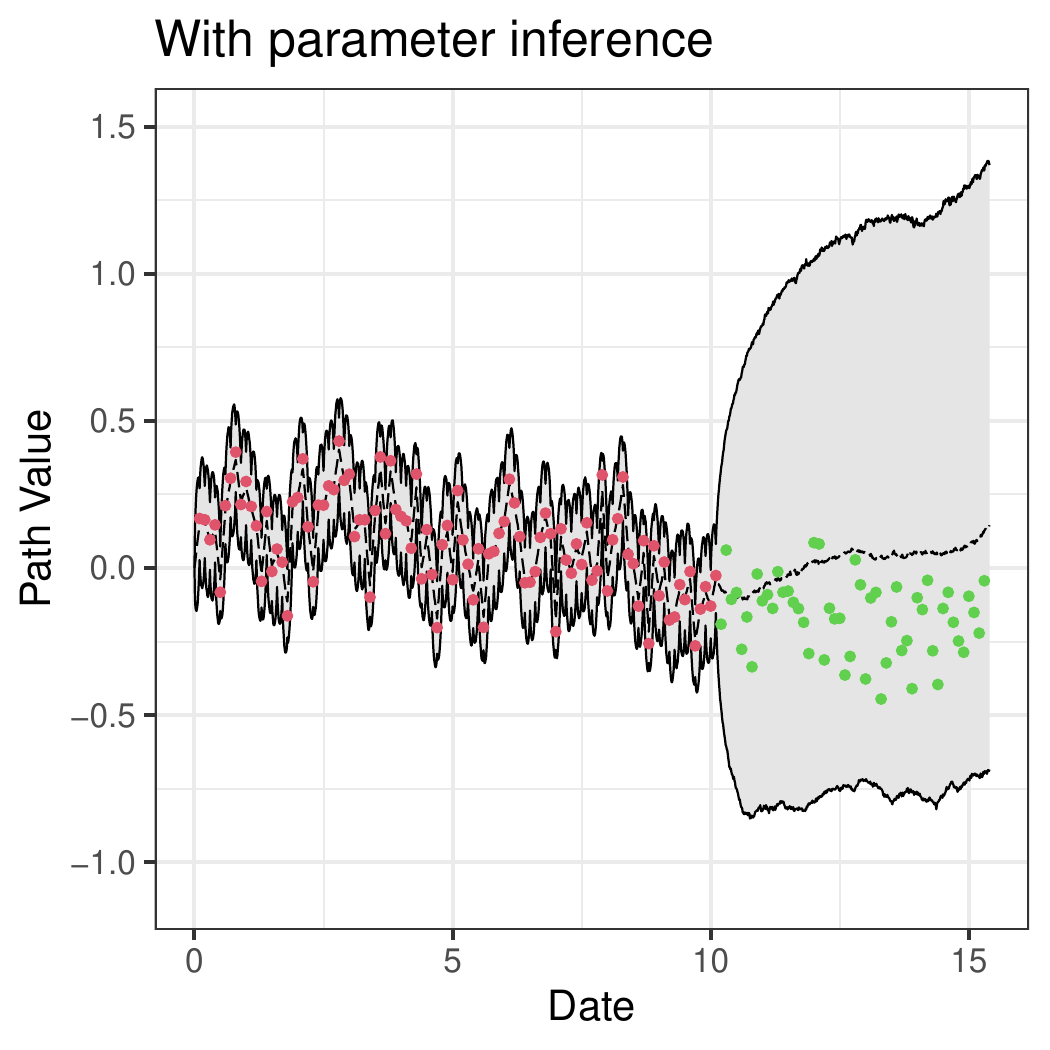}
 \label{fig:app_NYtemperature_pstpath}
\end{figure}

Additionally, we include trace plots of the middle time of the path $X_{T/2}$ and its autocorrelation function (ACF) plots to illustrate mixing in Figure \ref{fig:app_NYtemperature_tracetheta}. The fixed parameter case exhibits faster mixing compared to the parameter inference case.

\begin{figure}
 \centering
\begin{minipage}[c]{0.34\textwidth}
   \includegraphics[width=\linewidth]{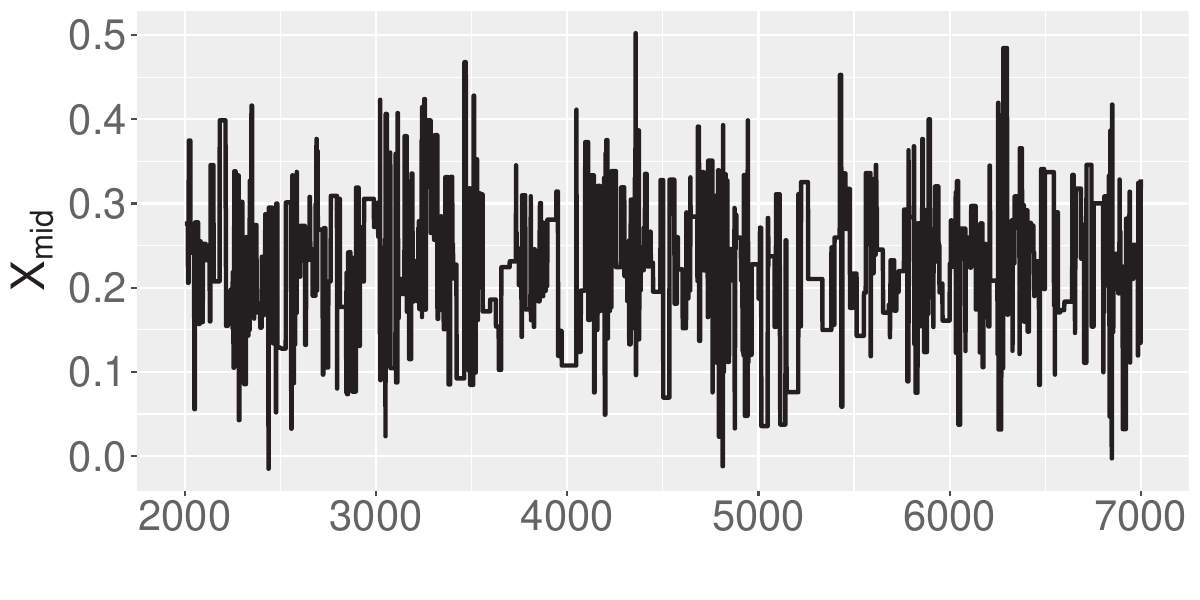}
\end{minipage}
\begin{minipage}[c]{0.34\textwidth}
  \includegraphics[width=\linewidth]{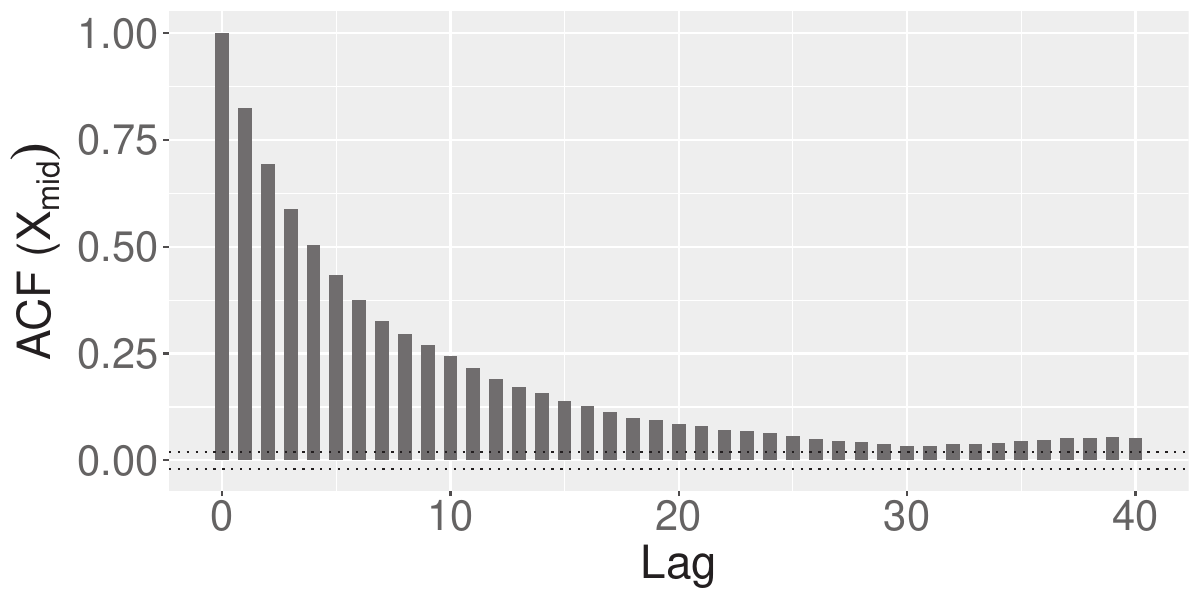}
\end{minipage}
\begin{minipage}[c]{0.34\textwidth}
   \includegraphics[width=\linewidth]{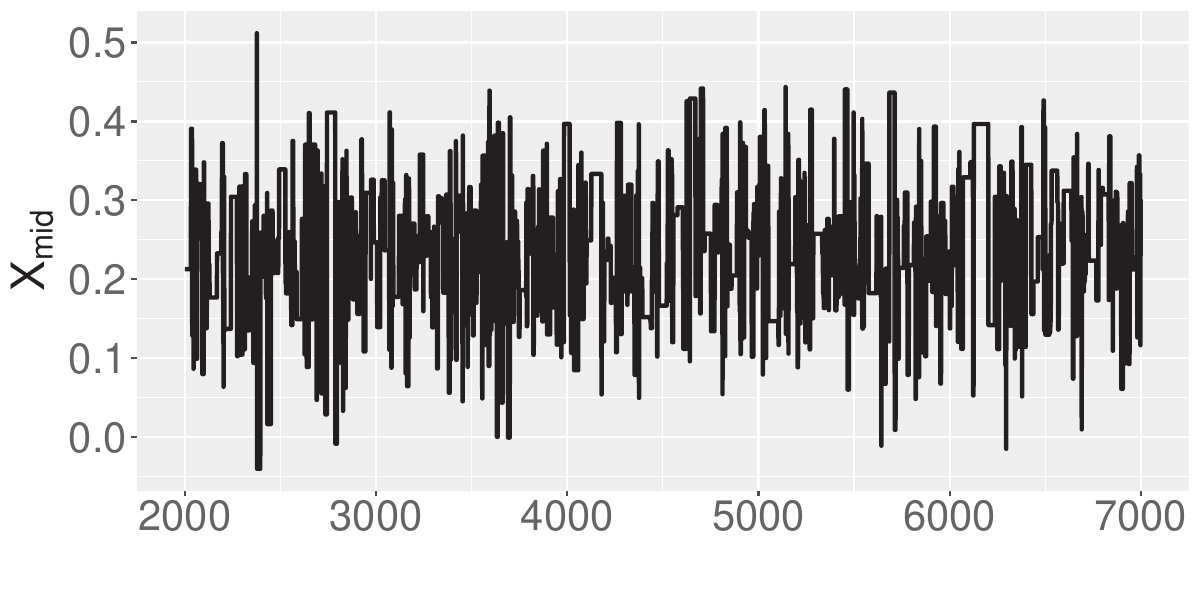}
\end{minipage}
\begin{minipage}[c]{0.34\textwidth}
   \includegraphics[width=\linewidth]{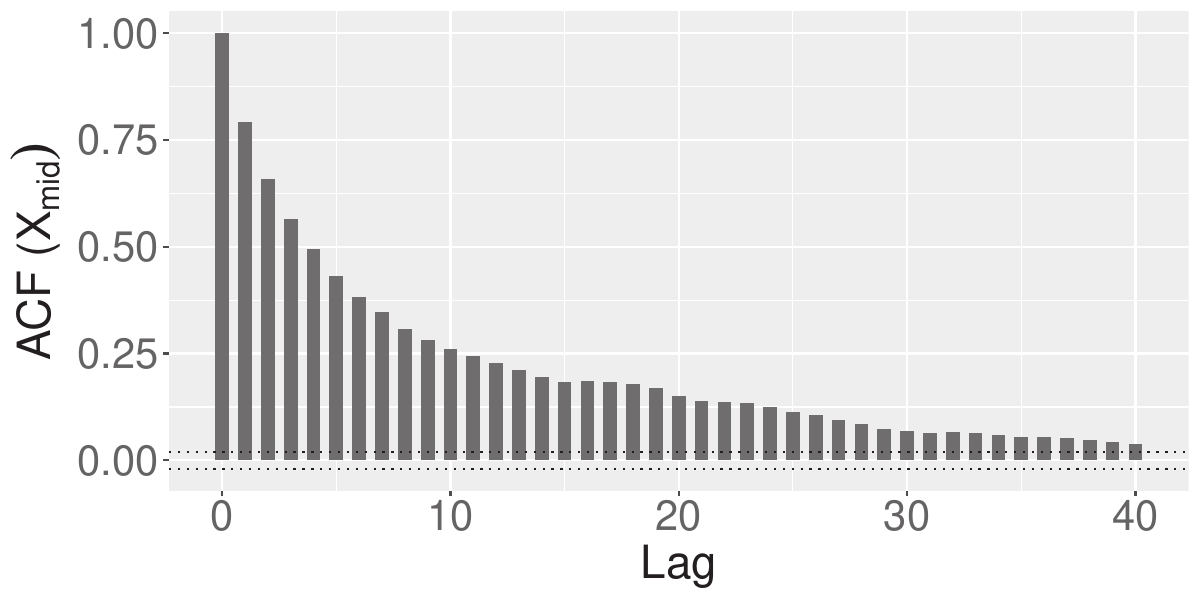}
\end{minipage}
\begin{minipage}[c]{0.34\textwidth}
  \includegraphics[width=\linewidth]{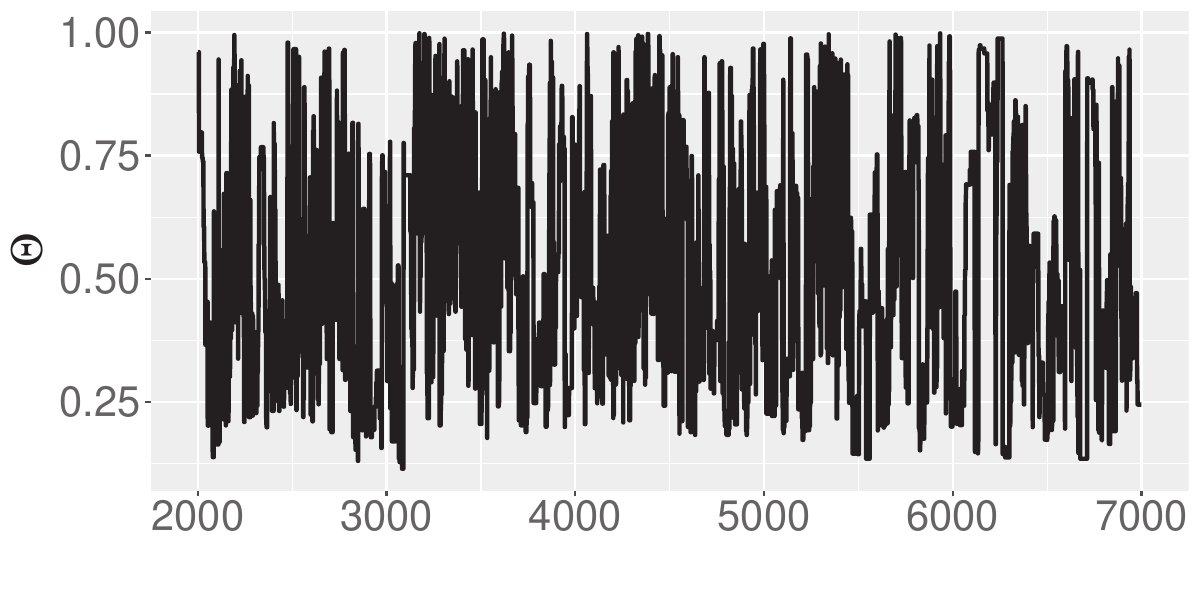}
\end{minipage}
\begin{minipage}[c]{0.34\textwidth}
   \includegraphics[width=\linewidth]{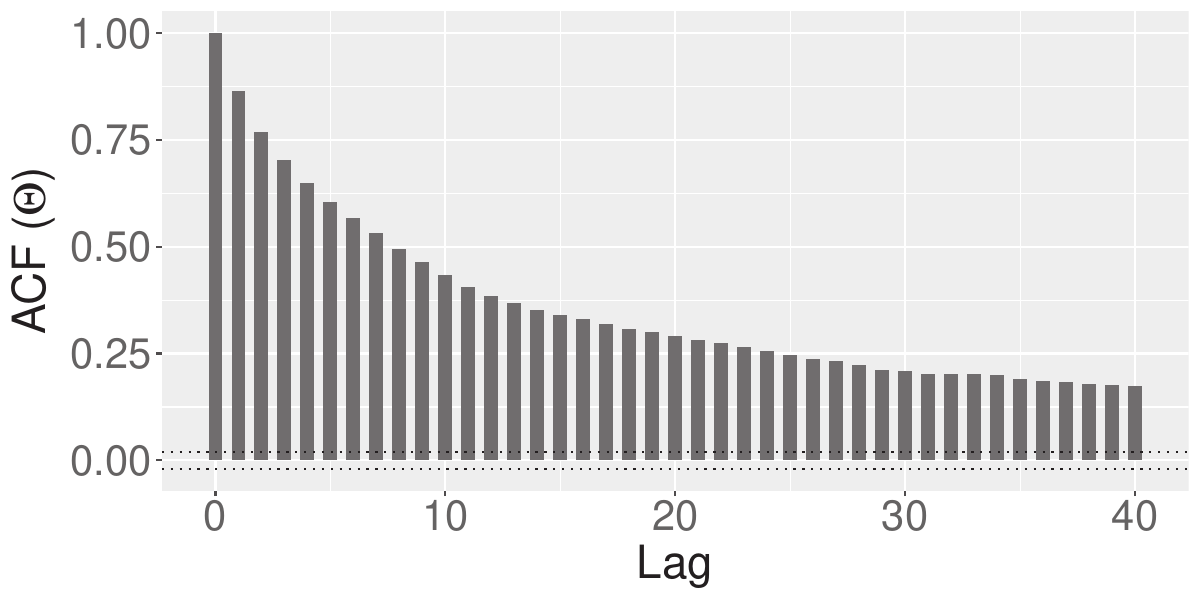}
\end{minipage}
\begin{minipage}[c]{0.34\textwidth}
  \includegraphics[width=\linewidth]{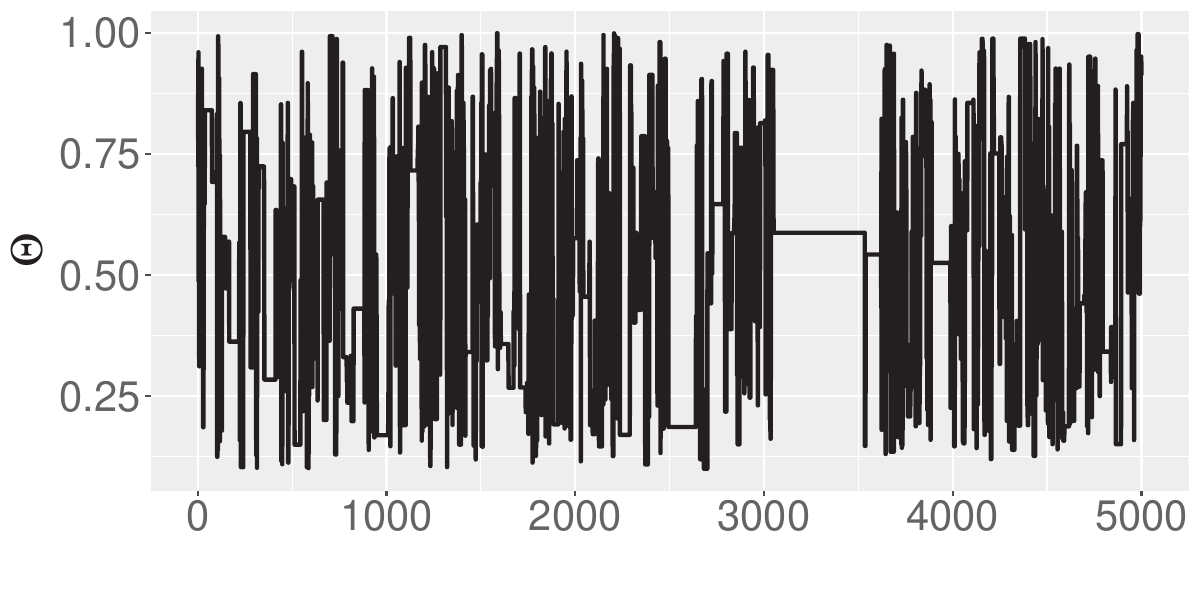}
\end{minipage}
\begin{minipage}[c]{0.34\textwidth}
   \includegraphics[width=\linewidth]{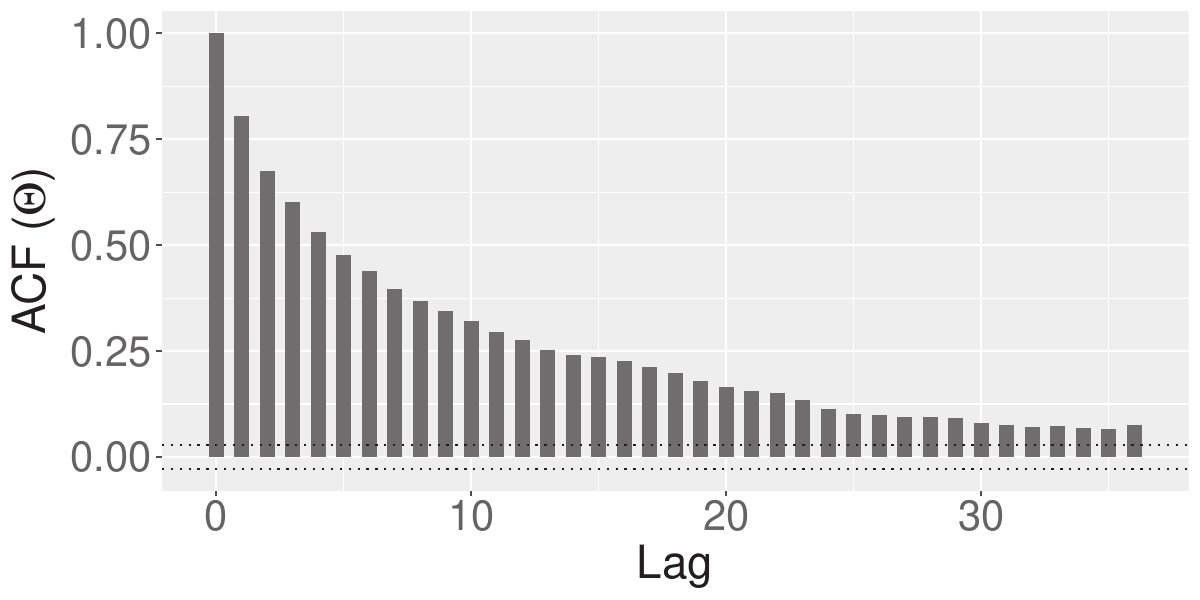}
\end{minipage}
\begin{minipage}[c]{0.05\textwidth}
$\quad$
\end{minipage}
\begin{minipage}[c]{0.5\textwidth}
\caption{Parameter inference for EA2 example 2 under Gaussian noise with \(\sigma_N = 0.1\). The top two rows show the trace and autocorrelation plots of $\theta$ (left) and $X_{mid}$, the SDE path value at time $T/2$. The bottom row shows the posterior distribution over the parameter $\theta$.}
 \label{fig:app_NYtemperature_tracetheta}
\end{minipage}
\begin{minipage}[c]{0.34\textwidth}
   \includegraphics[width=\linewidth]{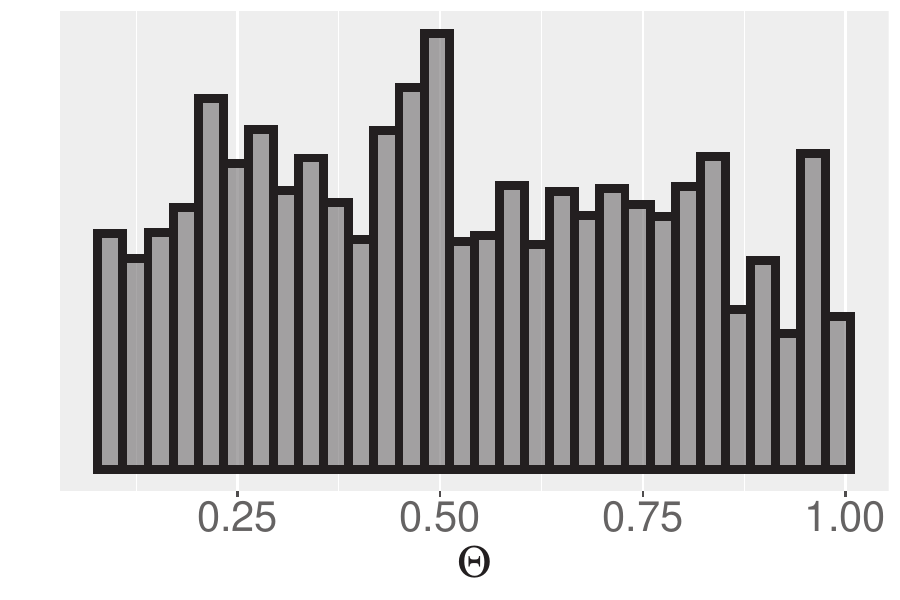}
\end{minipage}
\end{figure}

\clearpage

\section{Derivation of SDE functionals}
\label{apdix:SDE}

This section derives the functionals associated with the SDEs considered in this paper. These quantities are required both for classifying the class of diffusion and for implementing the corresponding algorithm. %

\subsection{EA2 example (positive drift function)}
\label{apdix:sde1}

Consider the following SDE with parameters $\theta = (p,q)$, where $p > 0$ and $q > 0$:
\begin{equation*}
    dX_t = p \cdot \exp(-qX_t) \, dt + dW_t
\end{equation*}
The drift function $\alpha_{\theta}(x) = p \cdot \exp(-qx)$ is strictly positive for all $x$, implying that the diffusion exhibits an upward tendency on average. For this SDE, we compute $\frac{1}{2}(\alpha_\theta^2(x) + \alpha_\theta'(x)) = \frac{1}{2} \left(p\cdot \exp(-qx)-\frac{q}{2}\right)^2 -\frac{q^2}{8}$. It follows that $\alpha_\theta^{\downarrow} = -q^2/8$ and $\phi_{\theta} (x) = \frac{1}{2} \left( p \cdot \exp(-qx) - \frac{q}{2} \right)^2$. The function $\phi_{\theta}(x)$ therefore satisfies the EA2 condition with $\underset{x\rightarrow-\infty}{\lim}\phi_{\theta}(x) = \infty$ for all $p,q > 0$. In this paper, we set the EA2 Poisson rate function as
\begin{equation*}
    M_{\theta}(x) = \begin{cases} 
\phi_{\theta}(x) & \text{if } x \leq \frac{1}{q} \log(p/q), \\
\frac{q^2}{8} & \text{if } x > \frac{1}{q} \log(p/q).
\end{cases}
\end{equation*}
Finally, straightforward calculations yield the end-point bias function of the $h$-biased Brownian bridge: $h_{\theta}(X_T| X_0) \propto \exp\left(-\frac{p}{q}\cdot\exp(-q\cdot X_T) - \frac{(X_T-X_0)^2}{2T}\right)$.

\subsection{EA2 example 2 (oscillatory drift function)}
\label{apdix:sde2}

We consider the diffusion parameterized by $\theta = (p,q)$ with $p > 0$, $q > 0$:
\begin{equation*}
    dX_t = - p X_t \exp(-q X_t) \, dt + dW_t
\end{equation*}

The drift $\alpha_\theta(x) = - p x \exp(-q x)$ is positive for $x < 0$ and negative for $x > 0$, causing the diffusion to move upward and downward, respectively. This oscillatory behavior is qualitatively different from that of the former EA2 example and makes this case more challenging. The function $\phi_{\theta}(x)$ is given by $\phi_{\theta}(x) = \frac{1}{2} \left\{ \left( p x \exp(-q x) + \frac{q}{2} \right)^2 - p \exp(-q x) - \frac{q^2}{4} \right\} - \alpha_\theta^{\downarrow}$, where $\alpha_\theta^{\downarrow} = \underset{x}{\min}\ \frac{1}{2}\left(\alpha_\theta^2(x)+\alpha'_\theta(x)\right)$ admits no closed-form expression. Thus, a closed-form expression for the EA2 Poisson rate $M_\theta(x)$ is also unavailable. We therefore use a numerical optimizer to obtain $\alpha_\theta^{\downarrow}$ and $M_\theta(x)$. The bias function is given by $h_{\theta}(X_T|X_0) \propto \exp\left\{ \frac{p}{q}\left(X_T+\frac{1}{q}\right)\cdot \exp(-q X_T) - \frac{(X_T-X_0)^2}{2T}\right\}$.

\subsection{EA3 example (double-well potential Langevin process)}
\label{apdix:sde3}

We consider the following SDE with parameter $\theta = (p,q)$ where $p>0$ and $q>0$:
\begin{equation*}
    dX_{t} = (-pX_{t}^{3}+qX_{t})dt + dW_{t},
\end{equation*}
a member of the \emph{double-well potential Langevin process}. For this SDE, one can verify $\alpha_{\theta}^{\downarrow} = -\frac{q}{2}-\frac{\sqrt{q^{2}+9p}}{3} + \frac{q^{3}}{27p} - \frac{q^{2}}{27p}\sqrt{q^{2}+9p}$
and $\phi_{\theta}(x) = \frac{1}{2}\left[p^{2}x^{6} - 2pqx^{4}+(q^{2}-3p)x^{2} + q\right] - \alpha_{\theta}^{\downarrow}$, so that $\phi_{\theta}(x)$ satisfies the EA3 condition. Recall that under the EA3 framework, $\phi_{\theta}(x)$ is unbounded above on either side of the real line. Thus, for each Brownian bridge segment, we work with a layer $\layer = (\Llayer,\Ulayer)$ that contains the path range $(\Xmin,\Xmax)$. Given a layer, we define the rate of Poisson process $M_{\theta}(\layer)$ as $M_{\theta}(\layer) = \sup_{x\in [\Llayer,\Ulayer]} \phi_{\theta}(x)$. In this example, $\phi_{\theta}(x)$ is a sixth-order polynomial, hence the supremum over a closed interval is attained either at the endpoints or at interior stationary points that are local maxima:
$
    M_{\theta}(\layer) =  \underset{x\in\mathcal{C}}{\sup}\ \phi_{\theta}(x)$,
where
\begin{equation*}
    \mathcal{C} = \{\Llayer,\Ulayer\} \cup \left(\{0\} \cap [\Llayer,\Ulayer]\right) \cup \begin{cases}
        \{\pm \zeta_{\theta}\} \cap [\Llayer,\Ulayer], & \text{if }q^{2} > 3p,\\
        \emptyset, & \text{otherwise.}
    \end{cases}
\end{equation*}
and
$
    \zeta_{\theta} = \left(\frac{2q-\sqrt{q^{2}+9p}}{3p}\right)^{1/2}
$
denotes the positive stationary point of $\phi_{\theta}(\cdot)$, corresponding to a local maximum when $q^{2} > 3p$. Finally, a simple calculus gives the bias function $h$:
$
    h_{\theta}(X_{T}|X_{0}) \propto \exp\left(-\frac{1}{4}pX_{T}^{4}+\frac{1}{2}qX_{T}^{2}-\frac{(X_{T}-X_{0})^{2}}{2T}\right).
$

\subsection{Parameter inference (CIR process)}
\label{apdix:sde4}

Let $\theta = (p,q,\sigma)$ with $p > 0$, $q > 0$, $\sigma > 0$, and consider the diffusion
$$
    dV_t = p(q - V_t)dt + \sigma\sqrt{V_t}dW_t.
$$
Applying the transform $X_t = 2\sqrt{V_t}/\sigma$ yields a unit diffusion coefficient process
\begin{align*}
    dX_t &= \underbrace{\left\{\frac{1}{X_t}\left(\frac{2pq}{\sigma^2}-\frac{1}{2}\right)-\frac{pX_t}{2}\right\}}_{:=\alpha_{\theta}(X_t)}dt + dW_t.
\end{align*}

Define the degree $d = 4pq/\sigma^2$ and assume that $d \geq 3$. Under this condition, the diffusion $X_t$ satisfies the Feller condition ($d > 2$) and therefore remains strictly positive. In addition, $\frac{1}{2}\left(\alpha_{\theta}^2(x) + \alpha_{\theta}^{\prime}(x)\right) = \left\{\left(\frac{2pq}{\sigma^2}-1\right)^2 - \frac{1}{4}\right\}\frac{1}{2x^2} + \frac{p^2}{8}x^2 - \frac{p^2 q}{\sigma^2}$ is bounded below on $(0,\infty)$; thus the diffusion belongs to the EA3 class \citep{reutenauer2008}. In particular, writing the expression in terms of $d$ and applying the AM--GM inequality yields
$    \alpha_{\theta}^{\downarrow} = \frac{p}{4}\left\{\sqrt{(d-1)(d-3)}-d\right\},
$
and thus
$
    \phi_{\theta}(x) = \left\{\left(\frac{2pq}{\sigma^2}-1\right)^2 - \frac{1}{4}\right\}\frac{1}{2x^2} + \frac{p^2}{8}x^2 - \frac{p}{4}\left\{\sqrt{(d-1)(d-3)}\right\}.
$ Finally, the bias function is given by $
    h_{\theta}(X_{T}|X_{0}) \propto \exp\left\{\left(\frac{2pq}{\sigma^2}-\frac{1}{2}\right)\log(X_T) - \frac{p}{4}X_T^2-\frac{(X_{T}-X_{0})^{2}}{2T}\right\}.
$

\end{document}